\documentclass[submission,copyright,creativecommons]{eptcs}

\usepackage{iftex}

\usepackage{subcaption}

  \usepackage[english]{babel}
  \usepackage[T1]{fontenc}        % Recommended with pdflatex

\title{Counterfactual Causality for Reachability and Safety based on Distance Functions}
\author{Julie Parreaux
\institute{Aix Marseille Univ, CNRS, LIS, Marseille, France}
\email{julie.parreaux@univ-amu.fr}
\and
Jakob Piribauer
\institute{Technische Universit\"at Dresden, Germany \\ Technische Universit\"at M\"unchen, Germany}
\email{jakob.piribauer@tu-dresden.de}
\and
Christel Baier
\institute{Technische Universit\"at Dresden, Germany}
\email{christel.baier@tu-dresden.de}
}

\usepackage[dvipsnames]{xcolor}
\usepackage{etex}
\usepackage{wrapfig}
\usepackage{graphicx}
\usepackage{amssymb,amsmath,amsthm}

\usepackage[mathscr]{eucal}
\usepackage{verbatim}
\usepackage{comment}
\usepackage{mathtools}
\usepackage{amsfonts}
\usepackage{nicefrac}
\usepackage{color}
\usepackage{graphicx}
\usepackage{float}
\usepackage{tikz}
\usepackage{xspace} 
\usepackage{pgf}
\usetikzlibrary{arrows,automata}

\usepackage[usenames,dvipsnames]{pstricks}
\usepackage{epsfig}
\usepackage{pst-grad} % For gradients
\usepackage{pst-plot} % For axes
\usepackage{pstricks-add}
\usepackage[linesnumbered,ruled]{algorithm2e}
\usepackage{hyperref}
\usepackage{cleveref}
\usepackage[backgroundcolor=magenta!10!white]{todonotes}
\usepackage{enumitem}

\usepackage{wasysym}
\usepackage{multirow}

\usetikzlibrary{automata,positioning,shadows}
\usetikzlibrary{calc, shapes, backgrounds,arrows}
\usetikzlibrary{chains, decorations.pathmorphing}
\usetikzlibrary{positioning,decorations.pathreplacing}
\usetikzlibrary{shapes.misc}

\tikzset{every loop/.style={looseness=7}, >=latex}
\tikzset{every picture/.style={>=latex}}
\tikzstyle{PlayerReach}=[draw,circle,minimum size=6mm,inner sep=1.5pt]
\tikzstyle{PlayerSafe}=[draw,rectangle,minimum size=6mm,inner sep=1.5pt]
\tikzstyle{Player}=[draw,diamond,minimum size=7mm,inner sep=1.5pt]
\tikzstyle{target}=[draw=none,rectangle, minimum size=1mm]
\tikzstyle{PlayerReachmin}=[draw,circle, minimum size=4mm,inner sep=0pt]
\tikzstyle{PlayerSafemin}=[draw,rectangle,minimum size=4mm,inner sep=0pt]
\tikzstyle{cause} =[draw=blue,line width = 0.5mm]

\tikzstyle{strat} =[minimum width=0.1cm,line width=0.01mm,draw=none]
\tikzstyle{vecArrow} = [decoration={markings,mark=at position
	1 with {\arrow[scale=.5,>=latex]{>}}}%,postaction={decorate}
]

\usepackage{thmtools, thm-restate}
\newtheorem{theorem}{Theorem}
\newtheorem{proposition}[theorem]{Proposition}
\newtheorem{lemma}[theorem]{Lemma}

\theoremstyle{definition}
\newtheorem{example}[theorem]{Example}
\newtheorem{definition}[theorem]{Definition}
\newtheorem{remark}[theorem]{Remark}

\definecolor{darkgreen}{rgb}{0.1, 0.5, 0.1}
\definecolor{RZ}{HTML}{385803}
\colorlet{simonColor}{Yellow!30!white}
\colorlet{jakobColor}{Green!30!white}

\newcommand{\eqdef}{\,\stackrel{\mathclap{\tiny\mbox{def}}}{=}\,}
\newcommand{\markend}{\hfill ${ \lrcorner}$}

\newcommand{\T}{\mathcal{T}}

\newcommand{\cT}{\T}

\newcommand{\sinit}{s_{\mathit{init}}}
\newcommand{\wgt}{\mathit{wgt}}
\newcommand{\AP}{\mathsf{AP}}

\newcommand{\dpref}{d_{\mathit{pref}}}
\newcommand{\dhamm}{d_{\mathit{Hamm}}}
\newcommand{\dghamm}{d_{\mathit{gHamm}}}

\newcommand{\dlev}{d_{\mathit{Lev}}}
\newcommand{\dstrat}[1]{d^{#1}}

\newcommand{\dist}{\mathit{dist}}
\newcommand{\dHH}{\ensuremath{d^*}\xspace}

\newcommand{\game}{\ensuremath{\mathcal{G}}\xspace}
\newcommand{\Pl}{\ensuremath{\Pi}\xspace}
\newcommand{\ReachPl}{\ensuremath{\mathsf{Reach}}\xspace}
\newcommand{\SafePl}{\ensuremath{\mathsf{Safe}}\xspace}
\newcommand{\MinPl}{\ensuremath{\mathsf{Min}}\xspace}
\newcommand{\MaxPl}{\ensuremath{\mathsf{Max}}\xspace}

\newcommand{\Locs}{\ensuremath{V}\xspace}
\newcommand{\LocsReach}{\ensuremath{\Locs_{\ReachPl}}\xspace}
\newcommand{\LocsSafe}{\ensuremath{\Locs_{\SafePl}}\xspace}
\newcommand{\LocsMin}{\ensuremath{\widetilde{\Locs}_{\MinPl}}\xspace}
\newcommand{\LocsMax}{\ensuremath{\widetilde{\Locs}_{\MaxPl}}\xspace}
\newcommand{\LocsT}{\ensuremath{\Locs_\textsl{Eff}}\xspace}
\newcommand{\LocsNegT}{\ensuremath{\Locs_{\neg\textsl{Eff}}}\xspace}
\newcommand{\loc}{\ensuremath{v}\xspace}
\newcommand{\locinit}{\ensuremath{v_{\textsf{i}}}\xspace}
\newcommand{\locT}{\ensuremath{\loc_\textsl{Eff}}\xspace}
\newcommand{\locNegT}{\ensuremath{\loc_{\neg\textsl{Eff}}}\xspace}

\newcommand{\Trans}{\ensuremath{\Delta}\xspace}
\newcommand{\trans}{\ensuremath{\delta}\xspace}

\newcommand{\weight}{\ensuremath{\textsf{wt}}\xspace}
\newcommand{\sCost}{\ensuremath{\textsf{max-cost}}\xspace}
\newcommand{\cost}{\ensuremath{\textsf{cost}}\xspace}

\newcommand{\loosestrategy}{\ensuremath{\sigma}\xspace}
\newcommand{\winstrategy}{\ensuremath{\tau}\xspace}
\newcommand{\strategy}{\ensuremath{\mu}\xspace}

\newcommand{\play}{\ensuremath{\xi}\xspace}
\newcommand{\looseplay}{\ensuremath{\pi}\xspace}
\newcommand{\winplay}{\ensuremath{\rho}\xspace}

\newcommand{\outcomes}{\mathsf{Play}}

\newcommand{\Z}{\mathbb{Z}}
\newcommand{\N}{\mathbb{N}}

\begin{document}
\maketitle
\paragraph*{Funding:}{\small{This work was partly funded by DFG Grant 389792660 as part of TRR 248 (Foundations of Perspicuous Software Systems), the Cluster of Excellence EXC 2050/1 (CeTI, project ID 390696704, as part of Germany’s Excellence Strategy), and  the DFG projects BA-1679/11-1 and BA-1679/12-1.}}
\vspace{24pt}

% !TEX root = ../main.tex

\begin{abstract}
Investigations of causality in operational systems aim at providing human-understandable explanations of \emph{why} a system behaves as it does.
There is, in particular, a demand to explain what went wrong on a given counterexample execution that shows that a system does not satisfy a given specification. 
%\todo[inline]{reformulate, complete}
%The counterfactuality principle defines the relation between a cause and its effect by stating that the effect would not have occurred if the cause had not occurred.
To this end, this paper investigates a notion of counterfactual causality in transition systems based on Stalnaker's and Lewis' semantics of counterfactuals in terms of most similar possible worlds and introduces a novel corresponding notion of counterfactual causality in two-player games. 
 Using distance functions between paths in transition systems to capture the similarity of executions, this notion defines whether reaching a certain set of states is a cause for the fact that a given execution of a system  satisfies an undesirable reachability or safety property. Similarly, using distance functions between memoryless strategies in reachability and safety games, 
 it is defined whether reaching a set of states is a cause for the fact that a given strategy for the player under investigation is losing.

The contribution of the paper is two-fold:
In transition systems, it is shown that counterfactual causality can be checked in polynomial time for three prominent distance functions between paths. 
In two-player games, 
the introduced notion of counterfactual causality is shown to be checkable in polynomial time for two natural distance functions between memoryless strategies.
Further,  a notion of  explanation that  can be extracted from a counterfactual cause and that pinpoints  changes to be made to the given strategy in order to transform it into a winning strategy  is defined. For the two distance functions under consideration, the problem to decide whether such an explanation imposes only minimal necessary changes to the given strategy with respect to the used distance function turns out to be coNP-complete and not to be solvable in polynomial time if P is not equal to NP, respectively.

\end{abstract}

\section{Introduction}

%\todo[inline]{to include: 
%
%-
%clear explanation of forward vs backward view. forward necessary causes vs backward counterfactual causes
%
%-
%Halpern-Pearl: minimal distance corresponds to minimal amount of interventions (explained in more detail in Section 3)
%}
%

Modern software and hardware systems have reached a level of complexity that makes it impossible for humans to assess whether a system behaves as intended without tools  tailored for this task.
To tackle this problem, automated verification techniques have been developed. \emph{Model checking} is one prominent such technique: A model-checking algorithm takes  a mathematical model of the system under investigation and a formal specification of the intended behavior and  determines whether all possible executions of the model satisfy the specification.
While the results of a model-checking algorithm provide guarantees on the correctness of a system or affirm the presence of an error, 
their usefulness  is, nevertheless, limited as they do not provide a human-understandable explanation of the  behavior of the system.

To provide additional information on \emph{why} the system behaves as it does,
certificates witnessing the result of the model-checking procedure, in particular counterexample traces in case of a negative result, have been studied extensively (see, e.g., \cite{ClarkeGMZ95,MaPn95,CGP99,Namjoshi01}). Due to the potentially still enormous size of counterexample traces and other certificates, a line of research has emerged that tries
to distill  comprehensible explications of what {causes} the system to behave as it does using formalizations of   \emph{causality} (see, e.g., \cite{Pearl09,Peters2017,ICALP21}).

\paragraph*{Forward- and backward-looking causality}

There are two fundamentally different types of notions of causality: \emph{forward-looking} and \emph{backward-looking} notions \cite{vandePoel2011}.
In the context of operational system models, forward-looking causality describes general causal relations between events that might happen along some possible executions. Backward-looking causality, on the other hand, addresses the causal relation between events along a given execution of the system model. 
This distinction is captured in more general contexts by the distinction between
\emph{type-level} causality addressing general causal dependencies between events that might happen when looking forward in a world model, and \emph{token-level} or \emph{actual} causality, corresponding to the backward view, that addresses causes for a particular event that actually happened  (see, e.g., \cite{Halpern15}). 

Notions of \emph{necessary} causality are typically  forward-looking: A necessary cause $C$ for an effect $E$ is an event that occurs on every execution that exhibits the effect $E$ (see, e.g.,  \cite{Baier2022Causality}, and for a philosophical analysis of necessity in causes \cite{Mackie65}).
The backward view naturally arises when the task is to explain what went wrong after an undesired effect has been observed. In the verification context, the backward view is natural for explaining counterexamples, see e.g. \cite{Zeller02,BallNR03,GroveV03,RenieresR03,WYIG06,WangAKGSL2013}. Most of these techniques rely on the \emph{counterfactuality} principle, which has been originally studied in philosophy \cite{Hume1739,Hume1748,Stalnaker1968,Lewis1973,Lewis1973counterfactuals} 
and formalized mathematically by Halpern and Pearl \cite{HalpernP2001,HalpernP04,HalpernP05,Halpern15}. 
Intuitively, counterfactual causality requires that 
the {effect} would not have happened, if the {cause} had not occurred, in combination with some minimality constraints for causes.
The most prominent account for the semantics of the involved counterfactual implication is provided 
   by Stalnaker and Lewis 
\cite{Stalnaker1968,Lewis1973,Lewis1973counterfactuals} 
in terms of closest, i.e.,  most similar, possible worlds. 
The statement ``if the cause $C$ had not occurred, then the effect $E$ would not have occurred'' holds true if in the worlds that are most similar to the actual world 
and in which $C$ did not occur, $E$ also did not occur. 
Interpreting executions of a system as possible worlds,
the actual world is an execution $\pi$ where both the effect $E$ and its
counterfactual cause $C$ occur, while the effect $E$ does not occur 
in alternative executions that
are as similar as possible to $\pi$ and that do not exhibit~$C$.

%\textcolor{gray}{ICH DENKE, DAS BSP K\"ONNTE MAN WEGLASSEN.
%While this might sound similar to a forward-looking necessary cause on first sight, the two notions are quite different. E.g., say we take our bike to work and do not want to get wet. Rain is now not a necessary cause for getting wet -- we might fall off our bike into a lake and get wet.
%However, when we arrive at work and did not fall into a lake, but it rained and we got wet, the rain is a counterfactual cause:
%The statement ``If it had not rained, we would not have gotten wet'' is now correct. The counterfactual implication does not talk about all possible other ways in which things could have evolved.}
%

%We will also illustrate the difference between forward- and backward-looking causality 
% in the  example below after explaining in more detail how we define counterfactual causality in transition systems and reachability games.

For a more detailed discussion on the distinction between forward- and backward looking causality and related concepts for responsibility, we refer the reader, e.g., to \cite{vandePoel2011,YazdanpanahD16,YazdanpanahDJAL19,BaierFM21,ICALP21}.

%

%
%For a more detailed discussion on this matter
%For the related concept
%As described for notions of responsibility in \cite{vandePoel2011,YazdanpanahD16,YazdanpanahDJAL19,Baier2022Causality}, 
%
%forward vs backward responsibility
%
%ICALP invited overview \cite{ICALP21}
%
%necessary causes \cite{Baier2022Causality}
%
%IJCAI 21 \cite{BaierFM21}
%

\paragraph*{Defining counterfactual causality in transition systems and reachability games}

%
%%
%In this work, we provide definitions of counterfactual causes in transition systems and reachability games based on the semantics of the counterfactual implication provided 
%   by Stalnaker and Lewis 
%\cite{Stalnaker1968,Lewis1973,Lewis1973counterfactuals} in terms of closest, i.e.,  most similar, possible worlds. 
%The statement ``if $C$ had not occurred, then $E$ would not have occurred'' holds true if in the worlds that are most similar to the actual world and in which $C$ did not occur $E$ also did not occur. 

To define our backward-looking notion of counterfactual causality in transition systems, we follow an approach similar to the one by Groce et al \cite{groce2006error} who presented a Stalnaker-Lewis-style formalization of  counterfactual dependence of events using distance functions.
We consider the case where effects are reachability or safety properties and causes are sets of states.
To illustrate the idea, let
 $\T$ be a transition system and let  $E$ and $C$  be disjoint sets of states of $\T$ indicating a reachability effect and a potential cause, respectively. 
Consider an execution $\pi$ that reaches the effect set and the potential cause set.
We  employ the counterfactual reading of causality by Stalnaker and Lewis by viewing executions as possible worlds  using a similarity metric  $d$ on paths:
Reaching $C$ was a cause for $\pi$ to reach $E$ if all paths $\zeta$, that do not reach $C$ and that are most similar to $\pi$ according to $d$ among all paths with this property, satisfy $\zeta\vDash \Box \neg E$, i.e.,  they do not reach $E$.
So, we first determine the minimal similarity-distance $d_{\min}=\min \{ d(\pi,\zeta) \mid \zeta\vDash \Box \neg C\}$ from $\pi$ to a path $\zeta$ that does not reach $C$.
Then, we check whether all paths that do not reach $C$ and have similarity-distance $d_{\min}$ to $\pi$ do not reach $E$: Do all 
$
\zeta\in \{\zeta^\prime\mid d(\pi,\zeta^\prime)=d_{\min}$  and $\zeta^\prime\vDash \Box \neg C\}$  satisfy $\Box \neg E$? 
If the answer is yes, it is the case that ``if $C$ had not occurred, then $E$ would not have occurred'' and so $C$ is a counterfactual cause for $E$ on $\pi$.

\begin{example}
\label{ex:traffic}
Consider the following distance function on paths in a labeled transition system $\cT$ with states $S$ and a labeling function $L\colon S\to \mathcal{A}$ for a set of labels $\mathcal{A}$:
For paths $\pi=s_0,s_1,\dots$ and $\pi^\prime=t_0,t_1,\dots$, we define 
$
\dist(\pi,\pi^\prime) = | \{n\in \mathbb{N} \mid L(s_n)\not=L(t_n) \}|$.
So, paths are more similar if their traces  differ at fewer positions. To determine whether $C$ is a cause for $E$ on $\pi$, we  first determine what the least number  $n_{\min}$ of changes to the state labels  of $\pi$ is  to obtain a path $\zeta$ that does not reach $C$. Then, we have to check whether all paths differing from $\pi$ in $n_{\min}$  labels and  not reaching  $C$ do  not reach $E$. If this is the case, $C$ is a counterfactual cause for $E$ on $\pi$ with respect to  $\dist$.

	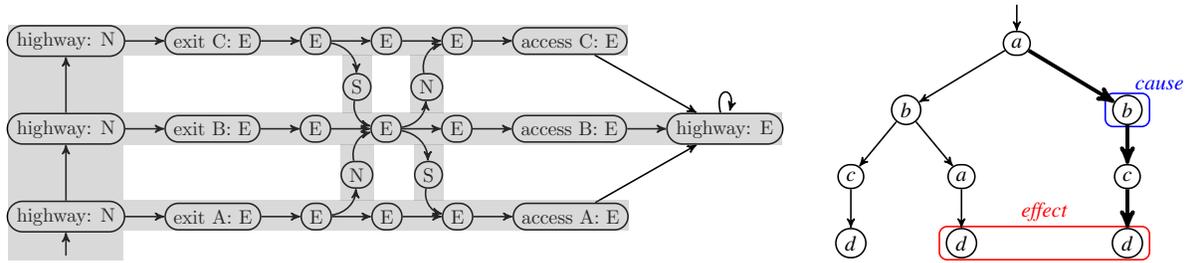
\begin{figure}[t]
	\resizebox{.65\textwidth}{!}{
		\includegraphics{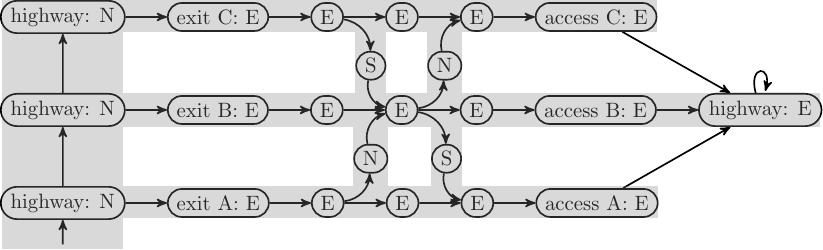}}	
		\hspace{12pt}			
		\resizebox{.3\textwidth}{!}{
		\begin{tikzpicture}
		[scale=1,->,>=stealth',auto ,node distance=0.5cm, thick]
		\tikzstyle{r}=[thin,draw=black,rectangle]
		
		\node[scale=1, rounded rectangle, draw] (start) {$a$};
		\draw[<-] (start) --++(0,+0.7);
		\node[scale=1, rounded rectangle, draw, below=.7 of start,xshift=-2cm] (step0) {$b$};
		\node[scale=1, rounded rectangle, draw, below=.7 of start,xshift=2cm] (step1) {$b$};
		
		\node[scale=1, rounded rectangle, draw, below=.7 of step0,xshift=-1cm] (step00) {$c$};
		\node[scale=1, rounded rectangle, draw, below=.7 of step0,xshift=1cm] (step01) {$a$};
		
		\node[scale=1, rounded rectangle, draw, below=.7 of step00] (step000) {$d$};
		%\node[scale=1, rounded rectangle, draw, below=.7 of step00,xshift=.5cm] (step001) {};
		
		%\node[scale=1, rounded rectangle, draw, below=.7 of step01] (step010) {};
		\node[scale=1, rounded rectangle, draw, below=.7 of step01] (step011) {$d$};

		\node[scale=1, rounded rectangle, draw, below=.7 of step1] (step10) {$c$};
		%\node[scale=1, rounded rectangle, draw, below=.7 of step1,xshift=1cm] (step11) {};

		\node[scale=1, rounded rectangle, draw, below=.7 of step10] (step100) {$d$};
		%\node[scale=1, rounded rectangle, draw, below=.7 of step10,xshift=.5cm] (step101) {};
		
		%\node[scale=1, rounded rectangle, draw, below=.7 of step11,xshift=-.5cm] (step110) {};
		%\node[scale=1, rounded rectangle, draw, below=.7 of step11,xshift=.5cm] (step111) {};

		\draw[color=blue,rounded corners] (step1)+(-.4,-.3)    rectangle  node[right=0,yshift=.45cm] {$\mathit{cause}$}  ($(step1)+(+0.4,+0.3)$);
		
			\draw[color=red,rounded corners] (step011)+(-.4,-.3)    rectangle  node[above=0.25] {$\mathit{effect}$}  ($(step100)+(+0.4,+0.3)$);

	\draw[color=black,->] (start) edge  (step0) ;
	\draw[color=black,->,line width=2pt] (start) edge  (step1) ;
	\draw[color=black,->] (step0) edge  (step00) ;
	\draw[color=black,->] (step0) edge  (step01) ;
	\draw[color=black,->,line width=2pt] (step1) edge  (step10) ;
	%\draw[color=black,->] (step1) edge  (step11) ;
	%\draw[color=black,->] (step00) edge  (step001) ;
	\draw[color=black,->] (step00) edge  (step000) ;
	
	\draw[color=black,->] (step01) edge  (step011) ;
	%\draw[color=black,->] (step01) edge  (step010) ;
	
	%\draw[color=black,->] (step10) edge  (step101) ;
	\draw[color=black,->,line width=2pt] (step10) edge  (step100) ;
	
	%\draw[color=black,->] (step11) edge  (step111) ;
	%\draw[color=black,->] (step11) edge  (step110) ;

%		
%		\draw[color=black,->] (continue) edge [loop right] node [pos=0.5,right] {\scriptsize$E$} (continue) ;
%		
%		
%		
%		

	\end{tikzpicture}}
\caption{On the left: example transition system modelling a traffic grid (Ex. \ref{ex:traffic}).  On the right: Example of a $\dhamm$-counterfactual cause that is not a $\dpref$-counterfactual cause (Ex. \ref{ex:dpref}).}
\label{fig:example_traffic}
		\end{figure}

Now, consider the example transition system $\cT$ modelling a  road system with a highway going north that has three  exits into a small town which can be left again on a highway heading east depicted in Fig. \ref{fig:example_traffic}.
		Each state is labeled with $N$, $E$, or $S$ for north, east, and south as indicated in Fig. \ref{fig:example_traffic} depending on the direction the cars move on the respective road.
		Say, an agent  traverses the system via the path $\pi$ with trace $NNE^\omega$, i.e., by taking exit B from the first highway and
then going eastwards straight through the town.		
Assume that there is a traffic jam on access B while the other access roads are free. The question is now whether taking exit B was a cause for being stuck in slow traffic  later on, i.e., for the effect $\{\text{access B}\}$.
		First, note here that the set $\{\text{exit B}\}$ is  not a forward-looking necessary cause for reaching $\{\text{access B}\}$. There are  paths through the  system that avoid $\{\text{exit B}\}$, but reach $\{\text{access B}\}$.
		
		However, given the fact that the agent traversed the town by going straight eastwards, it is reasonable to say that the agent would have reached a different access road if she had taken a different exit from the first highway.
		This is  reflected in the counterfactual definition using  $\dist$: There are two paths that do not reach exit B and whose traces differ from $\pi$ at only one position, namely the paths with trace $NE^\omega$ and $NNNE^\omega$. These paths do also not reach access B. So,  $\{\text{exit B}\}$
		is a counterfactual cause for $\{\text{access B}\}$ on $\pi$; if the agent had taken exit A or C, she would not have hit the low traffic flow at access B. 
%		
%		Note that if we use the metric $d_{\mathit{tree}}$, this is not the case: The paths not reaching exit B that have the longest common prefix with $\pi$ among all such paths are those that reach exit $C$. However, there is the path with action sequence $NNNEEESE^\omega$ that reaches the low traffic flow while using exit C. So, under this similarity metric, we cannot say that low traffic flow would have been avoided if exit B had not been chosen. 
%
\markend
\end{example}

%\todo{Explain a bit why we also want to study causality in games}
In the context of two-player reachability games, causality has  been used as a tool to solve games \cite{BaierCFFJS21}. 
In our work, we  focus on explaining why a certain strategy does not allow the player to win. 
More precisely, in reachability games between players 
%$\SafePl$ and $\ReachPl$ 
with a safety and the complementing reachability objective, respectively, we consider the situation where one of the players $\Pl$ has a winning strategy, but loses the game using a strategy $\sigma$. We introduce a notion of counterfactual causality that aims to provide insights into what is wrong with strategy $\sigma$ by
transferring the counterfactual definition using  distance functions $d$ on memoryless strategies. A set of states $C$ is said to be a $d$-counterfactual cause for the fact that $\sigma$ is losing if all memoryless strategies $\tau$, that make sure that $C$ is not reached and have minimal $d$-distance to $\sigma$ among all such strategies, are winning. Furthermore, we introduce \emph{counterfactual explanations} that specify minimally invasive changes of $\sigma$'s decisions required to turn $\sigma$ into a winning strategy.

\paragraph*{Contributions}
\begin{itemize}
\item
We show that $d$-counterfactual causal relationships in transition systems (defined as in \cite{groce2006error}) 
can be checked in polynomial time for the following three distance metrics $d$ (Sec. \ref{sub:check_TS}):
\begin{enumerate}
\item
the prefix distance: paths are more similar if their traces share a longer prefix.
\item
the Hamming distance that counts the positions at which  traces of paths differ.
\item
the Levenshtein distance that counts how many insertions, deletions, and substitutions are necessary to transform the trace of one path to the trace of another path.
\end{enumerate}
Furthermore, we show that the  notion of $d$-counterfactual causality for the Hamming distance is consistent with Halpern and Pearl's but-for causes \cite{HalpernP04,HalpernP05}.
%
%Furthermore, we briefly discuss the problem of finding good causes. For cardinality-minimal causes, a non-deterministic polynomial time upper bound is obvious.
%The distance that a cause $C$ enforces on the most similar executions avoiding the cause from the actual execution, however, turns out to be a questionable quality measure (Sec. \ref{sub:finding_TS}).
\item
In reachability games, we provide a generalization of this notion using similarity metrics on memoryless deterministic strategies. 
We show that for the Hausdorff lifting of the prefix distance on paths to a distance function on memoryless deterministic strategies, the resulting notion can be checked in polynomial time (Sec. \ref{sec:counterfactual_games}).
\item
We introduce a notion of \emph{counterfactual explanation} that can be computed from a counterfactual cause (Sec. \ref{sec:explain-strat}).
An explanation specifies where a non-winning  strategy needs to be changed. Of particular interest are $D$-minimal explanations that enforce only minimal necessary changes with respect to a distance function $D$ on strategies. For two distance functions related to the Hamming distance, we show that checking whether an explanation is minimal is coNP-complete and not in P if P$\not=$NP, respectively.
\end{itemize}
An overview of the complexity results can be found in Table~\ref{table:contribution}.

%\begin{table}[tbp]
%	\centering
%	\begin{tabular}{|c|l|c|l|}
%		\hline
%		& \multicolumn{2}{l|}{Distance $d$} & Checking $d$-counterfactual causality  \\
%		\hline
%		\multirow{3}{1.75cm}{\centering Transition system} & 
%		\multicolumn{2}{l|}{Prefix distance} & 
%		in P (Thm.~\ref{thm:check_dpref_TS}) \\
%		\cline{2-4}
%		& \multicolumn{2}{l|}{Hamming distance} & 
%		in P (Thm.~\ref{thm:checking_dhamm})  \\ 
%		\cline{2-4}
%		& \multicolumn{2}{l|}{Levenshtein distance} & 
%		in P (Thm.~\ref{thm:checking_Levenstein})  \\
%		\hline
%		\multirow{5}{*}{Games} & 
%		\multicolumn{2}{l|}{Hausdorff lifting $\dpref^H$ of the prefix distance} & 
%		in P (Thm.~\ref{thm:checking-pref-game}) \\
%			\cline{2-4}
%			&\multicolumn{2}{l}{} & \\
%		&\multicolumn{2}{l}{} & Checking $d$-minimality of explanations \\
%		\cline{2-4}
%		& \multicolumn{2}{l|}{Hamming strategy distance $\dhamm^s$} & 
%		{coNP-complete (Cor. \ref{cor:check-hamm-game})} \\
%		\cline{2-4}
%		& \multicolumn{2}{l|}{Hausdorff-inspired distance \dHH}  & {not in P if P$\not=$NP (Cor. \ref{cor:check-hamm-game})}\\
%		\hline
%	\end{tabular}
%	\caption{Overview of the complexity results.}
%	\label{table:contribution}
%\end{table}

\begin{table}[tbp]
	\centering
	\hspace{-24pt}
	\begin{subtable}[b]{.3\linewidth}
	\begin{tabular}{|c|c|}
		\hline
		distance $d$ & causality \\
		\hline
		\multirow{2}{*}{prefix} & in P  \\
		& (Thm.~\ref{thm:check_dpref_TS}) \\
		\hline
		\multirow{2}{*}{Hamming} & in P  \\ 
		& (Thm.~\ref{thm:checking_dhamm}) \\
		\hline
		\multirow{2}{*}{Levenshtein} & in P \\
		& (Thm.~\ref{thm:checking_Levenstein}) \\
		\hline
	\end{tabular}
	\end{subtable}
	\begin{subtable}[b]{.6\linewidth}
	\begin{tabular}{|c|c|c|}
		\hline
		distance $d$ & causality & explanations \\
		\hline
		Hausdorff lifting $\dpref^H$  & 
		in P & \\
		of the prefix distance & (Thm.~\ref{thm:checking-pref-game}) & \\
		\hline
		Hamming strategy & in P in acyclic games &
		coNP-complete  \\
		distance $\dhamm^s$ & (Thm.~\ref{thm:checking-hamm-game}) & (Cor. \ref{cor:check-hamm-game}) \\
		\hline
		Hausdorff-inspired  &  & not in P if P$\not=$NP \\
		distance \dHH & & (Cor. \ref{cor:check-hamm-game}) \\
		\hline
	\end{tabular}
	\end{subtable}
	\caption{Overview of the complexity results. On the left, the complexities of checking $d$-counterfactual 
		causality  in transition systems, and on the right, the complexities of checking $d$-counterfactual 
		causality and  $d$-minimality of explanations  in reachability games.}
	\label{table:contribution}
\end{table}

\paragraph*{Related work}

Ways to pinpoint the problematic steps in a counterexample trace by 
localizing errors  have widely  been studied  \cite{Zeller02,BallNR03,GroveV03,RenieresR03,WYIG06,WangAKGSL2013}.
For counterfactuality in transition systems, we follow the approach of \cite{groce2006error} with distance metrics. In contrast to the causes in this paper,  causes in \cite{groce2006error} are formulas in an expressive logic that can precisely talk about the valuation of variables after a certain number of steps. Further  \cite{groce2006error} is not concerned with checking causality, but with finding causes, which, due to the expressive type of causes, algorithmically boils down to finding executions avoiding the effect with a minimal distance to the given one.

Based on  counterfactuality, Halpern and Pearl \cite{HalpernP04,HalpernP05,Halpern15} provided an influential formalization of causality using structural equation models, which   has served as the basis for various notions of causality in the verification context  (see,e.g., \cite{BeerBCOT2012,Leiner-FischerL2013}).
A key ingredient is the notion of \emph{intervention} to provide a semantics for the counterfactual implication in Hume's definition of causality. 
 An intervention in a structural equation model sets a variable to a certain value by force, ignoring its dependencies on other variables, and evaluates the effects of this enforced change. In a sense, a minimal set of interventions to avoid a cause then leads to a most similar execution avoiding the cause. We will discuss the relations between our definition and the Halpern-Pearl definition in more detail in Section \ref{sec:HP-causality}.
 In \cite{BeerBCOT2012}, interventions are employed to counterexample traces in transition systems by allowing to flip  atomic propositions along a  trace. In contrast to our notion of counterfactual causes, this is tailored for complex linear time properties, but does not provide insights for reachability and safety. Furthermore, the flipping of atomic propositions can be seen as a change in the transition system while our definition  considers alternative executions without manipulating the system.
In~\cite{coenen2022temporal}, the Halpern-Pearl approach is applied to provide a counterfactual definition of 
causality in reactive systems. 
%Here causes are $\omega$-regular properties of an input sequence, while effects are $\omega$-regular properties of the output sequence. 
A distance partial order, namely the subset relation on sets of positions at which traces differ, is used to describe which interventions are acceptable as they constitute minimal changes  necessary to avoid the cause.  Checking causality is  shown to be decidable by a formulation as a hyperlogic model-checking problem.
Furthermore, notions of necessary and sufficient causes as sets of states in transition systems have been considered \cite{Baier2022Causality}. These do not rely on the counterfactuality principle and are of forward-looking nature.

%In reachability games, \cite{BaierCFFJS21} defines necessary and sufficient subgoals in the same spirit, which, nevertheless, serve a different purpose there.

%Recently, notions of causality based on the probability-raising principle that a cause increases the probability of its effects \cite{Reichenbach1956} have  been investigated for stochastic operational systems.
%Kleinberg and Mishra \cite{KleinbergM2009,KleinbergM2010,Kleinberg2011} used this principle in Markov chains. This work has subsequently sparked research on further probabilistic causality notions in Markov chains and Markov decision processes \cite{ZiemekPFJB22,BaierFPZ22}.
%Probabilistic notions of causality have also been formulated in terms of  hyperproperties~\cite{abraham2018hyperpctl,DimFinkbeinerTorfah-probHyper-MDP-ATVA2020}.

%Besides notions of causality, several other concepts have been studied with the goal of providing understandable explanations of the behavior of a system, e.g., \emph{coverage estimation}   \cite{HoskoteKHZ99,ChocklerKKV01}, or \emph{vacuity detection} \cite{BeerBER97,KupfermanV99}. 

We are not aware of formalisations of causality in game structures. The related concept of responsibility, has been investigated in 
multi-agent models \cite{YazdanpanahD16,YazdanpanahDJAL19}. Notions of forward and backward responsibility of players in multi-player
game structures with acyclic tree-like  arena have been studied \cite{BaierFM21}.

For a detailed overview of work on causality and related concepts in operational models, we refer the reader to the survey articles \cite{Chockler16,ICALP21}.

\section{Preliminaries}
\label{sec:prelim}

We briefly present notions we use and our notation. For details, see \cite{BK08,GradelTW-02}.

\noindent \textit{Transition systems.}
A transition system is a tuple $\cT=(S,\sinit,\to,L)$ where $S$ is a finite set of states, $\sinit\in S$ is an initial state, $\to\subseteq S\times S$ is a transition relation and $L\colon S\to 2^{\AP}$ is a labeling function where $\AP$ is a set of atomic propositions.
A path in a transition system is a finite or infinite sequence of states $s_0 s_1\dots$ such that $s_0=\sinit$ and, for all suitable indices $i$, there is a transition from $s_i$ to $s_{i+1}$, i.e., $(s_i,s_{i+1})\in \to$. Given a path $\pi=s_0 s_1\dots$, we denote its trace $L(s_0)L(s_1)\dots$ by $L(\pi)$.
If there are no outgoing transitions from a state, we call the state \emph{terminal}.

\noindent \textit{Computation tree logic (CTL).} The branching-time logic CTL consists of state formulas that are evaluated at states in a transition system  formed by $\Phi ::= \top \mid a \mid \Phi\land \Phi \mid \neg \Phi \mid \exists \varphi \mid \forall \varphi$ 
where $a\in \AP$ is an atomic proposition 
 and path formulas evaluated on paths formed by $\varphi::= \bigcirc \Phi \mid \Phi \mathrm{U} \Phi$. The semantics of the temporal operators in path formulas is as usual. We use the abbreviations $\lozenge \Phi$ for $\top \mathrm{U} \Phi$ and $\Box \Phi = \neg \lozenge\neg \Phi$ and also allow sets of states $T$ in the place of state formulas. The semantics of $\exists \varphi$ are that there exists a path starting in the state at which the formula is evaluated that satisfies $\varphi$; $\forall \varphi$ is defined dually to that as usual. Model checking of CTL-formulas can be done in polynomial time. For details, see \cite{BK08}.

\noindent \textit{Reachability games.}
	A \emph{reachability game} is a tuple $\game = (\Locs, \locinit, \Trans)$ where  
	$\Locs = \LocsReach \uplus \LocsSafe \uplus \LocsT$ is the set of 
	vertices shared between players \ReachPl and 
	\SafePl, and  some target vertices $\LocsT$ ($\textsl{Eff}$ for effect).
	$\locinit \in \Locs \setminus \LocsT$ is the initial vertex and 
	$\Trans \subseteq \Locs \times \Locs$ is the set of edges.
We denote by 
$\Trans(\loc)$ the set of edges from \loc.
W.l.o.g., we assume that target vertices are terminal states, 
i.e. for all vertices $\loc \in \LocsT$, $\Trans(\loc) = \emptyset$. A 
\emph{finite play} is a finite sequence of vertices 
$\looseplay = \loc_0\loc_1\cdots \loc_k\in \Locs^*$ such
that for all $0\leq i<k$, $(\loc_i,\loc_{i+1})\in \Trans$. A \emph{play} is either 
a finite play ending in a target vertex, or an infinite sequence of vertices 
such that every finite prefix is a finite play.
Transition systems can be viewed as one-player games.

A \emph{strategy} for \ReachPl in a reachability game \game is a mapping
$\loosestrategy \colon \Locs^* \LocsReach \to \Locs$. A play or finite
play $\looseplay = \loc_0\loc_1\cdots$ is a \emph{\loosestrategy-play} if
for all~$k$ with $\loc_k \in \LocsReach$, we have 
$\loosestrategy(\loc_0\cdots\loc_k) = \loc_{k+1}$. A strategy \loosestrategy is 
an \emph{MD-strategy} (for memoryless deterministic) if for all finite plays \play and $\play'$ with the same last 
vertex, we have that $\loosestrategy(\play)=\loosestrategy(\play')$. In this paper, we mainly use MD-strategies and write 
$\loosestrategy(\loc_k)$ instead of $\loosestrategy(\loc_0\cdots\loc_k)$ for MD-strategies $\loosestrategy$.
 Moreover, under 
a (partial) MD-strategy \loosestrategy, we define the 
\emph{reachability game under \loosestrategy}, denoted by $\game^{\loosestrategy} = 
(\Locs, \locinit, \Trans^{\loosestrategy})$, by removing edges not chosen by $\sigma$, i.e.,
$
\Trans^{\loosestrategy} = \Trans \setminus 
	\{ (\loc, \loc')\in \Trans \mid  \loc \in \LocsReach\text{ and $\sigma(\loc)$ is defined and $\sigma(\loc)\not=(\loc, \loc')$} \}
$.
When \loosestrategy is completely defined, $\game^{\loosestrategy}$ is a 
transition system. Finally, a strategy is \emph{winning} if all 
\loosestrategy-plays starting in \locinit  end in a target vertex. Analogous definitions 
apply to \SafePl.
	In reachability games, either \ReachPl or \SafePl wins with an MD-strategy. This winning strategy can be computed in 
	polynomial time (see, e.g., \cite{GradelTW-02}).

\noindent \textit{Distance function.}
A \emph{distance function} on a set $A$ is a function $d\colon A \times A \to \mathbb{R}_{\geq 0}\cup\{\infty\}$ such that 
  $d(x,x)=0$  for all $x \in A$ and
  $d(x,y)=d(y,x)$  for all $x,y \in A$.
It is called a \emph{pseudo-metric} if additionally
 $d(x,y)+d(y,z)\geq d(x,z)$  for all $x,y,z \in A$, and
 a \emph{metric} if further 
 $d(x,y)=0$ holds iff $x=y$  for all $x,y \in A$.

\section{Counterfactual causes in transition systems}
\label{sec:counterfactual_in_TS}

In this section, we introduce the backward-looking notion of counterfactual causes in transition systems using distance functions (Section \ref{sub:definition_TS}). Afterwards, we prove that the definition can be checked in polynomial time for three well-known distance functions (Section~\ref{sub:check_TS}). Finally, we illustrate similarities between our notion of counterfactual causality to the definition of causality by Halpern and Pearl (Sec \ref{sec:HP-causality}).
Proofs omitted here can be found in Appendix \ref{app:counterfactual_in_TS}.

\subsection{Definition}
\label{sub:definition_TS}
The effects we consider  are reachability or safety properties $\Phi=\lozenge E$ or $\Phi=\Box \neg E$ for a set of states $E$. As the behavior of the system after $E$ has been seen is not relevant for these properties, we assume that $E$ consists of terminal states.
\begin{definition}[$d$-counterfactual cause in transition systems]
	\label{def:conterfactual_TS}
Let $\cT$ be a transition system and let $d$ be a distance function on the set of maximal paths of $\cT$.
Let $E$ be a set of terminal  states and let $C$ be a set of states disjoint from $E$. Let $\Phi=\lozenge E$ or $\Phi=\Box \neg E$.
Given a maximal path $\pi$ that visits $C$ and satisfies $\Phi$ in $\cT$, we say that $C$ is a \emph{$d$-counterfactual cause for $\Phi$ on $\pi$} if 
\begin{enumerate}
\item there is a maximal path $\rho$ in $\cT$ that does not visit $C$, and
\item all maximal paths $\rho$ with $\rho\vDash \Box \neg C$ with minimal distance to $\pi$  do not satisfy $\Phi$.
In other words,  all maximal paths $\rho$ with $\rho\vDash \Box \neg C$ such that $d(\pi,\rho)\leq d(\pi,\rho^\prime)$ for all $\rho^\prime$ with
$\rho^\prime \vDash \Box \neg C$ satisfy $\rho\vDash \neg \Phi$. 
\end{enumerate}
\end{definition}

The choice of the similarity distance $d$ of course heavily influences the notion of $d$-counterfactual cause.
In this paper, we will instantiate the definition with three distance functions that are among the most prominent distance functions between traces (or words). 
An experimental investigation to clarify in which situations what kind of distance functions leads to a desirable notion of causality, however, remains as future work.

 \noindent\textit{Prefix metrics $\dpref^{\AP}$ and $\dpref$:} given two paths $\pi$ and $\rho$, let $n(\pi,\rho)$ be the length of the longest common prefix of their traces $L(\pi)$ and $L(\rho)$. Then, 
$
\dpref^{\AP}(\pi,\rho)\eqdef 2^{-n(\pi,\rho)}$.
We can also define the distance on paths instead of traces, which will be used later on: $\dpref(\pi,\rho)\eqdef 2^{-m(\pi,\rho)}$ where 
 $m(\pi,\rho)$ is the length of the longest common prefix of $\pi$ and $\rho$ as paths.
This can be seen as a special case of $\dpref^{\AP}$ if we assume that all states have a unique label.

The prefix metric measures similarity in a temporal way saying that executions are more similar if they initially agree 
for a longer period of time. If no further structure of the transition system or meaning of the labels is known, this distance function might be a 
reasonable  choice for counterfactual causality.

 \noindent\textit{Hamming distance $\dhamm$:} Given two words $w=w_0\dots w_n$ and $v=v_0\dots v_n$ of the same length,
 we define
$
\dhamm(w,v)  \eqdef   |\{0\leq i \leq n \mid w_i \not= v_i \} |$.
For two maximal paths $\pi$ and $\rho$ of the same length in a transition system $\cT$ with labeling function $L$, we define 
$\dhamm(\pi,\rho) \eqdef \dhamm(L(\pi),L(\rho))$. So, the distance between two paths is the Hamming distance of their traces. 
%Note that  $\dhamm$ is a metric on words, but only a pseudo-metric on maximal paths.]

The Hamming distance seems to be a reasonable measure if a system naturally proceeds through different layers, e.g., if a counter is increased in each step. Then, traces are viewed to be more similar if they agree on more layers. The temporal order of these layers, however, does not play a role.

 \noindent\textit{Levenshtein distance $\dlev$ \cite{levenshtein1966binary}:}
Given two words $w=w_0\dots w_n$ and $v=v_0\dots v_m$, the Levenshtein distance is defined as the minimal number of editing operations needed to produce $v$ from $w$ where the allowed operations are insertion of a letter, deletion of a letter, and substitution of a letter by a different letter. 
Formally, we define $\dlev$ in terms of \emph{edit sequences}.
Let $\Sigma$ be an alphabet and $v,w\in \Sigma^\ast\cup\Sigma^\omega$ be two words over $\Sigma$.
The \emph{edit alphabet} for $\Sigma$ is defined as 
$
\Gamma \eqdef  (\Sigma\cup\{\varepsilon\})^2 \setminus \{(\varepsilon,\varepsilon)\}
$
where $\varepsilon$ is a fresh symbol.
An edit sequence for $v$ and $w$ is now a word $\gamma\in \Gamma^\ast\cup\Gamma^\omega$ such that the projection of $\gamma$ onto the first component results in $v$ when all $\varepsilon$s are removed and the projection of $\gamma$ onto the second component results in $w$ when all $\varepsilon$s are removed.
E.g., let $\Sigma=\{a,b,c\}$, $v=abbc$ and $w=accbc$. One  edit sequence is 
$
\gamma= (a,a)(b,c)(\varepsilon,c)(b,b)(c,c)$.
The weight  of an edit sequence  $\gamma=\gamma_1\gamma_2\dots$  is defined as
$
\wgt(\gamma)=| \{i \mid \gamma_i\not= (\sigma,\sigma)\text{ for all $\sigma \in \Sigma$}\}|$.
Then, for all words $v\in \Sigma^\ast\cup\Sigma^\omega$ and $w\in \Sigma^\ast\cup\Sigma^\omega$, we define
$
\dlev(v,w)=\min\{\wgt(\gamma)\mid \gamma \text{ is an edit sequence for $v$ and $w$}\}$.
%
%
%Formally, the Levenshtein distance can be defined recursively as follows where we use the notation $w_{[i:j]}$ to denote the subword $w_i\dots w_j$:
%\[
%\dlev(w_{[1:n]}, v_{[1:m]}) \eqdef 
%\begin{cases}
%n& \text{if $m=0$,}\\
%m& \text{if $n=0$,}\\
%\dlev(w_{[2:n]}, v_{[2:m]}) & \text{if $w_1=v_1$,}\\
%1+ \min 
%\begin{Bmatrix}
%\dlev(w_{[1:n]},v_{[2:m]}), \\
%\dlev(w_{[2:n]},v_{[1:m]}), \\
% \dlev(w_{[2:n]},v_{[2:m]})
%\end{Bmatrix}  
%&\text{otherwise.}
%\end{cases}
%\]
Again, we obtain a pseudo-metric on paths  via the Levenshtein metric on traces.
%
% Note that if exactly one of the words $v$ and $w$ is infinite, their $\dlev$-distance is infinite.

The Levenshtein distance might be particularly useful if 
labels model actions that are taken. Two executions that are obtained by sequences of actions that only differ by 
inserting or leaving out some actions, but otherwise using the same actions, are considered to be similar in this case.

\begin{example}\label{ex:dpref}
% An example for the notion of counterfactual cause resulting from the Hamming distance  was given in Example \ref{ex:traffic} in the introduction (note, however, that there labels were placed on transitions and not on states and that infinite executions were considered).
 Let us  illustrate  counterfactual causality for the prefix metric $\dpref$ and the Hamming distance $\dhamm$.
%\begin{figure}[t]
%	\begin{center}
%	\resizebox{.45\textwidth}{!}{
%		\includegraphics{figures/tree_example.pdf}}
%		\end{center}
%		\vspace{-12pt}
%		\caption{Example of a $\dpref$-counterfactual cause}
%		\label{fig:example_prefix}
%		\end{figure}
%
Consider the  transition system depicted in Figure \ref{fig:example_traffic}. A path $\pi$ as indicated by the bold arrows on the right via the potential cause to the effect  has been taken:
This is not a $\dpref$-counterfactual cause on $\pi$: The most similar paths to $\pi$ that do not reach $\mathit{cause}$  are both paths that move to the left initially. As one of these paths reaches $\mathit{effect}$, the set $\mathit{cause}$ is not a $\dpref$-counterfactual cause for reaching $\mathit{effect}$. 

Considering the distance function $\dhamm$ with the labels of the states as in Figure \ref{fig:example_traffic}, we get a different result:
The trace of $\pi$ is $abcd$. The paths that avoid the potential cause have traces $abcd$ and $abad$, respectively. So, the most similar path avoiding $\mathit{cause}$ is the path on the left with trace $abcd$ that also avoids $\mathit{effect}$. So, $\mathit{cause}$  is a $\dhamm$-counterfactual cause on $\pi$ for $\lozenge\mathit{effect}$. Intuitively, this can be understood as saying if the system had avoided $\mathit{cause}$ but otherwise behaved (as similar as possible to) as it did in terms of the produced trace, the effect would not have occurred. 
In particular, if labels represent actions that have been chosen, this is  a reasonable reading of causality.
\markend
\end{example}

%\todo[inline]{Reviewer 4: "I am not entirely sure why the authors use the definitions based on closest paths and simple counterfactual reasoning rather than the framework of actual causality (that is also cited in the paper, [HP05] and [Hal15]). Fig 2 demonstrates a flaw in this approach, where s1 is not a cause for reaching the effect. Note also that the notion of counterfactual cause is used here somewhat confusingly, since the causes computed in the paper are neither necessary nor sufficient." }

%
%\begin{wrapfigure}{r}{5cm}
%	\begin{center}
%	\resizebox{.55\textwidth}{!}{
%		\includegraphics{figures/tree_example.pdf}}
%		\end{center}
%		\vspace{-12pt}
%		\caption{Example of a $\dpref$-counterfactual cause}
%		\label{fig:example_prefix}
%		\end{wrapfigure}
%
%

\subsection{Checking counterfactual causality in transition systems}
\label{sub:check_TS}

In this section, we provide algorithms to check $d$-counterfactual causality for the three distance functions $\dpref^{\AP}$, $\dhamm$, and $\dlev$.
For these algorithms, a maximal execution $\pi$ of the system has to be given. We assume that $\pi$ is a finite path ending in a terminal state.
The problem to find causes that are small or satisfy other desirable properties is not addressed in this paper and remains as future work.
We will briefly come back to this in the conclusions.

% More formally, we provide algorithms for the following resulting decision problem:
%
%\paragraph{Given:} A transition system $\cT$, a set of terminal states $E$ and a set of states $C$ disjoint from $E$, an execution $\pi$ reaching $C$ and $E$,  as well as one metric $d\in \{\dpref, \dhamm ,\dlev\}$.
%
%\paragraph{Task:} Decide whether $C$ is a $d$-counterfactual cause for $E$ on $\pi$.

\paragraph*{Prefix distance.}
First, we consider  $\dpref^{\AP}$-counterfactual causality and hence $\dpref$-counterfactual causality as a special case.

\begin{restatable}{theorem}{checkprefAPTS}
\label{thm:check_dpref_TS}
Let $\cT=(S,\sinit,\to,L)$ be a transition system, $E$ a set of terminal states, $C$ a set of states disjoint from $E$, 
and $\Phi=\lozenge E$ or $\Phi=\Box \neg E$.
Let $\pi=s_0 \dots s_n$ be an execution reaching $C$ and satisfying $\Phi$.
It is decidable in polynomial time whether $C$ is a $\dpref^{\AP}$-counterfactual cause for $\Phi$ on $\pi$.
\end{restatable}

\begin{proof}[Proof sketch]
The following algorithm solves the problem in polynomial time:
 First, we determine the last index $i$  s.t.  $C$ is not reached on any path with trace $L(s_0),\dots,L(s_i)$
and s.t. $C$ is avoidable from  some state that is reachable via a path with trace $L(s_0),\dots,L(s_i)$. In order to that,
we recursively construct sets $T_{j+1}$ of states that are reachable via paths with trace $L(s_0),\dots,L(s_{j+1})$
and check for all states $t\in T_{j+1}$ whether  $t\vDash \exists \Box \neg C$. If no such state exists, we have found the first index $j+1$ such that $C$ is not avoidable anymore after trace $L(s_0),\dots,L(s_{j+1})$; so  we have found $i=j$.
Now, we check whether $t \vDash \forall (\Phi \to \lozenge C)$ for all $t\in T_i$. If this is the case, $C$ is a $\dpref^{\AP}$-counterfactual cause for $E$ on $\pi$; otherwise, it is not.
%For the runtime and correctness of this algorithm, see Appendix \ref{app:counterfactual_in_TS}. 
\end{proof}

%\begin{proof}[Proof sketch]
%Instead of finding the last position $s_i$ in $\pi$ from which $C$ is still avoidable, we construct a sequence of sets $T_i$ that can be reached via a path with  trace $L(s_0)\dots L(s_i)$ and find the largest $i$ such that $C$ can be avoided from some state in $T_i$. Then, we check whether all states  $t\in T_i$ satisfy $t\vDash \forall (\Phi \to \lozenge C)$. If this is the case, $C$ is a $\dpref^{\AP}$-counterfactual cause; otherwise, it is not.  
%\end{proof}

\paragraph*{Hamming distance.}

The Hamming distance is only defined for words of the same length. We will hence first consider only transition systems in which all maximal paths have the same length. We can think of such transition systems as being structured in layers with indices $1$ to $k$ for some $k$. Transitions can then only move from a state on layer $i<k$ to a state on layer $i+1$.
Afterwards, we consider a simple generalization of the Hamming distance to words of different lengths.

\noindent\textit{Original Hamming distance.}
Let $\cT=(S,\sinit, \to, L)$ be a transition system  in which all maximal paths have the same length $k$. We annotate all states with the layer they are on: For each state $s\in S$, there is a unique length $n\leq k$ of all paths from $\sinit$ to $s$. We will say that state $s$ lies on layer $n$ in this case.
By our assumption that effect states are terminal, the states $E$ are all located on the last layer $k$. We assume furthermore that all effect states have the same labels.

\begin{restatable}{theorem}{checkhammTS}
\label{thm:checking_dhamm}
Let $\cT=(S,\sinit, \to, L)$ be a transition system  in which all maximal paths have the same length $k$. Let $E$ be a set of terminal states and
let $C\subseteq S$ be a set of states disjoint from $E$. 
Let $\Phi=\lozenge E$ or $\Phi=\Box \neg E$.
Let $\pi=s_0 \dots s_n$ be an execution reaching $C$ and satisfying $\Phi$.
It is decidable in polynomial time whether $C$ is a $\dhamm$-counter\-factual cause for $\Phi$ on $\pi$.
\end{restatable}

\begin{proof}[Proof sketch]
We sketch the proof for the case that $\Phi=\lozenge E$. 
 We  equip the states in $S$ with a weight function $\wgt\colon S\to \{0,1\}$ such that the $\dhamm$-distance of a path to $\pi$ is equal to the accumulated weight of that path. 
 A state $t$ on layer $i$ gets weight $1$ if its label is different to $L(s_i)$. Otherwise, it gets weight $0$.
Now, we can check  whether $C$ is a $\dhamm$-counterfactual cause, as follows: We remove all states in $C$ and compute a shortest (i.e., weight-minimal) path $\zeta$ to $E$ and a shortest path $\xi$ to any terminal state. If the weight of $\xi$ is lower than the weight of $\zeta$, the paths avoiding $C$ that are $\dhamm$-closest to $\pi$ do not reach $E$ and $C$ is a $\dhamm$-counterfactual cause for $\lozenge E$ on $\pi$; otherwise, it is not.
\end{proof}

\begin{remark}\label{rem:weightedHamming}
The Hamming distance between paths could easily be extended to account for different levels of similarities between labels: Given a  similarity metric $d$ on the set of labels, one could define the distance between two paths $\pi=s_1\dots s_k$ and $\rho=t_1\dots  t_k$ as $\dhamm^\prime(\pi,\rho) \eqdef \sum_{i=1}^k d(s_i,t_i)$. 
%The original Hamming distance would be obtained by setting the distance between two distinct labels to $1$.
 The algorithm  in the proof of Theorem \ref{thm:checking_dhamm} can now easily be adapted to this modified Hamming distance by defining the weight function on the transition system in the obvious way.
\end{remark}

\noindent\textit{Generalized Hamming distance.}
Of course, the assumption that all paths in a transition system have the same length that we used in the previous section is quite restrictive.
Hence, we now consider the following generalized version $\dghamm$ of the Hamming distance: For words $w=w_1\dots w_n$ and $v=v_1 \dots v_m$, we define
\[
\dghamm(w,v) \eqdef 
\begin{cases}
\dhamm(w,v_{[1:n]}) + (m{-}n) & \text{if $n\leq m$,} \\
\dhamm(w_{[1:m]},v) + (n{-}m)& \text{otherwise.} 
\end{cases}
\]
So $\dghamm$ takes a prefix of the longer word of the same length as the shorter word, computes the Hamming distance of the prefix and the shorter word, and adds the difference in length of the two words.

\begin{restatable}{theorem}{checkghammTS}
Let $\cT=(S,\sinit,\to,L)$ be a transition system, $E$ a set of terminal states, and $C$ a set of states disjoint from $E$.
Let $\Phi=\lozenge E$ or $\Phi=\Box \neg E$.
Let $\pi=s_0 \dots s_n$ be an execution reaching $C$ and satisfying $\Phi$.
It is decidable in polynomial time whether $C$ is a $\dghamm$-counter\-factual cause for $\Phi$ on $\pi$.
\end{restatable}

\begin{proof}[Proof sketch]
We adapt the proof of Theorem \ref{thm:checking_dhamm}:
We take $|\pi|$-many copies of the state space $S$ an let transitions lead from one copy to the next.
In the $i$th copy states with the same label as $s_i$ get weight $0$ and all other states get weight $1$.
Furthermore, we add transitions with weight $|\pi|-i$ from terminal states in a copy $i<|\pi|$ to the same state in the last copy to account for path that are shorter than $\pi$. 
The weight $|\pi|-i$ corresponds to the value added in the generalized Hamming distance when paths of different length are compared.
To account for paths longer than $\pi$, we furthermore allow transitions with weight $1$ within the last copy. These transitions are then taken until a terminal state is reached.
With these adaptations, the proof can be carried out analogously to  the proof of Theorem \ref{thm:checking_dhamm}.
\end{proof}

\paragraph*{Levenshtein distance.}
% Consequently, we, e.g., also say that $C$ is a $\dlev$-counterfactual cause for reaching a set $E$ of terminal states on a run $\pi$ in a transition system if all paths that do not reach $C$ are infinite. These paths with distance $\infty$ to $\pi$ are then the closest ones avoiding $C$ and they do not reach $E$ as they are infinite.
 The idea to check $\dlev$-counterfactual causality is to construct a weighted transition system to be able to check causality via the computation of shortest paths as for the Hamming distance. 
So, let $\cT=(S,\to,\sinit,L)$ be a transition system labeled by $L$ with symbols from  the alphabet $\Sigma=2^{\AP}$.
Let $E$ be a set of terminal  states and $C$ a set of states disjoint from $E$.
Let $\Phi=\lozenge E$ or $\Phi=\Box \neg E$.
Let $\pi=s_1\dots s_n$ be a maximal path reaching $C$ and satisfying $\Phi$.
The  transition system we construct in the sequel  contains transitions corresponding directly to the edit operations insertion, deletion and substitution. A path in the constructed transition system  then corresponds
to an edit sequence between the trace of $\pi$ and the trace of another path in $\cT$.
This construction  shares some similarities with the construction of Levenshtein automata \cite{schulz2002fast} that accept all words with a Levenshtein distance below a given constant $c$ from a fixed word $w$.

Now, we formally construct the new weighted transition system $\cT_{\dlev}^\pi$:
The state space of this transition system is $S\times \{1,\dots, n\}$ with the initial state $(\sinit,1)$. The labeling function  is not used.
In $\cT_{\dlev}^\pi$, we allow the following transitions  labelled with letters from the edit alphabet~$\Gamma$:
\begin{enumerate}
\item
 a transition from $(s,i)$ to $(t,i+1)$ labeled with $(L(s_{i+1}),L(t))$ for each $(s,t)\in \to$ and $i<n$, 
\item
 a transition from $(s,i)$ to $(t,i)$ labeled with $(\varepsilon,L(t))$ for each $(s,t)\in \to$ and $i\leq n$,
\item 
 a transition from $(s,i)$ to $(s,i+1)$ labeled with $(L(s_{i+1}), \varepsilon)$ for each $s\in S$ and $i<n$.
\end{enumerate}
Note that the terminal states in $\cT_{\dlev}^\pi$ are all contained in $S\times\{n\}$.
Any maximal path in $\cT_{\dlev}^\pi$ corresponds to a maximal path $\rho$ in $\cT$.
This path $\rho$ is obtained by moving from a state $s$ to a state $t$ in $\cT$ whenever a corresponding transition of type 1 or 2 is taken in $\cT_{\dlev}^\pi$.
Transitions of type 3 do not correspond to a step in $\cT$ and stay in the same state.

Furthermore, given a finite path $\tau$  in $\cT_{\dlev}^\pi$ and the corresponding path $\rho=t_1\dots t_k$  in $\cT$, the labels of the transitions of $\tau$ form an edit sequence
for the words $L(s_2)\dots L(s_n)$ and $L(t_2)\dots L(t_k)$. 
To see this, observe that, for each $i>1$, whenever the copy $S\times\{i\}$ is entered in $\cT_{\dlev}^\pi$, the label of the transition contains $L(s_i)$ in the first component; if a transition stays in a copy $S\times\{i\}$, the label contains $\varepsilon$ in the first component. So, the projection onto the first component of the labels of the transitions of $\tau$  is indeed $L(s_2)\dots L(s_n)$, potentially with $\varepsilon$s in between. In the second component, whenever a transition of type 1 or 2 is taken, the label is simply the label of the corresponding state in $\rho$. Transitions of type 3 have $\varepsilon$ in the second component of their label.
Note here that $\rho$ and $\pi$ both start in $\sinit$ and that we could hence add $(L(\sinit),L(\sinit))$ to the beginning of the  edit sequence to obtain an edit sequence for the full traces of $\tau$ and $\rho$.
Note that also for infinite paths $\tau=t_1t_2\dots$  in $\cT_{\dlev}^\pi$ the transition labels provide an edit sequence for the words
$L(s_2)\dots L(s_n)$ and $L(t_2)L(t_3)\dots$.
Vice versa,  a finite maximal path $\rho=t_1\dots t_k$  in $\cT$ together with an edit sequence $\gamma$ for $L(s_2)\dots L(s_n)$ and $L(t_2)\dots L(t_k)$  provides  a maximal path 
$\tau$ in 
$\cT_{\dlev}^\pi$: The occurrences of $\varepsilon$ in $\gamma$ dictate which type of transition to take while the path $\rho$ tells us which state to move to.
As $\gamma$ projected to the first component contains $L(s_2)\dots L(s_n)$ enriched with $\varepsilon$s exactly $n-1$ transitions of type 1 or 3 are taken in $\tau$ obtained in this way and we indeed reach the last copy $\cT\times\{n\}$. As $\rho$ ends in a terminal state $t_k$, we furthermore reach the terminal state $(t_k,n)$. 
Analogously, an infinite path $\rho$  in $\cT$ together with an edit sequence $\gamma$ for $L(\pi)$ and $L(\rho)$ yields an infinite path in $\cT_{\dlev}^\pi$.

Based on these observations, we equip $\cT_{\dlev}^\pi$ with a weight function $\wgt$ on transitions: Transitions labeled with $(\sigma,\sigma)$ for a $\sigma\in \Sigma$
get weight $0$, the remaining transitions get weight $1$.

\begin{restatable}{theorem}{checklevTS}
	\label{thm:checking_Levenstein}
Let $\cT=(S,\sinit,\to,L)$ be a transition system, $E$ a set of terminal states, and $C$ a set of states disjoint from $E$.
Let $\Phi=\lozenge E$ or $\Phi=\Box \neg E$.
Let $\pi=s_0 \dots s_n$ be an execution reaching $C$ and satisfying $\Phi$.
It is decidable in polynomial time whether $C$ is a $\dlev$-counterfactual cause for $\Phi$ on $\pi$.
\end{restatable}

\begin{proof}[Proof sketch]
With the construction of the weighted transition system $\cT_{\dlev}^\pi$ above, the check can be done via the computation of shortest paths as for the Hamming distance above.
\end{proof}

\subsection{Relation to Halpern-Pearl causality}
\label{sec:HP-causality}

In the sequel, we want to demonstrate how our definition of counterfactual causality relates to Halpern-Pearl-style definitions of causality in \emph{structural equation models} \cite{HalpernP04,HalpernP05,Halpern15}.
A structural equation model consists of variables $X_1,\dots, X_n$ with finite domains   that are governed by   equations 
$
X_{i}=f_i(X_1,\dots, X_{i-1}, C)
$
for all $i\leq n$. Here, $f_i$ is an arbitrary function for each $i$ and $C$ is an input parameter for the context. For our consideration, the context $C$ does not play a role and we will hence omit it in the sequel.
So,  the value of variable $X_{i}$ depends on the value of (some of) the variables with lower index and the dependency is captured by the function $f_i$.
 Halpern and Pearl  use  \emph{interventions}
to define causality for an effect $E$, which is a set of valuations of $X_1, \dots, X_n$. An intervention puts the value of a variable $X_{i}$ to some $\alpha$ that is different from $f_i(X_1,\dots,X_{i-1})$, i.e., disregarding the equation $f_i$. Afterwards, the subsequent variables are evaluated as usual or by further interventions. Halpern and Pearl define:
\begin{definition}
Let  $f_1,\dots, f_n$ over variables $X_1,\dots, X_n$  be a structural equation model as above and let $E$ be an effect set of valuations such that the valuation of $X_1,\dots, X_n$ obtained by the structural equation model belongs to $E$. A \emph{but-for-cause} is a minimal subset $X\subseteq\{X_1,\dots, X_n\}$ with the following property:
There are values $\alpha_x$ for $x\in X$ such that putting variables $x\in X$ to $\alpha_x$ by intervention leads to a valuation of $X_1,\dots,X_n$ not exhibiting the effect $E$. More precisely, letting $t_i$ be the valuation $[X_1=w_1,\dots,X_{i-1}=w_{i-1}]$, where
$
w_i =
 f_{i}(w_1,\dots,w_{i-1})$   
 if $X_i\not\in X$, and 
$w_i = \alpha_{X_i}$ if $X_i\in X$, 
we get that $t_{n+1}\not\in E$.
\end{definition}

In order to compare this to our notion of counterfactual causes, we view structural equation models as tree-like transition system $\cT$:
The nodes at level $i$ are valuations for the variables $X_1,\dots, X_{i-1}$. At each node $s$ at level $i$, two actions are available: The action $\mathsf{default}$ moves to the state on level $i+1$ where the valuation in $s$ is extended by setting $X_i$ to the value $f_i(X_1,\dots,X_{i-1})$ where the values for 
$X_1,\dots,X_{i-1}$ are taken from the valuation in $s$. The action $\mathsf{intervention}$ extends the valuation of $s$ by setting $X_i$ to any other value than the action $\mathsf{default}$. 
The labelling in $\cT$ assigns the label $\{\mathsf{intervention}\}$ to all states that are reached by the action $\mathsf{intervention}$. The remaining states reached via $\mathsf{default}$ and the initial state with the empty valuation get the label $\emptyset$.
Given an effect $E$ as a set of valuations, we interpret this as the corresponding set of leaf states in $\cT$. 
The default path $\pi$ that always chooses the action $\mathsf{default}$ corresponds to evaluating the equations in the structural equation model without interventions. 
 We can now capture but-for-causality with $\dhamm$-counterfactual causality along the default path $\pi$ if all variables  are Boolean:

\begin{restatable}{proposition}{butforcausality}
\label{prop:butfor}
Let  $f_1,\dots, f_n$ over Boolean variables $X_1,\dots, X_n$  be a structural equation model and let $E$ be an effect set of valuations.
Let $X$ be a but-for-cause for $E$. Let 
$C_X$ be the set of all nodes in the transition system $\cT$ which are reached by a $\mathsf{default}$-transition for a variable $x\in X$.
Then, $C_X$ is a $\dhamm$-counterfactual cause for $E$ in $\cT$ on the default path $\pi$.
\end{restatable}

For non-Boolean variables, the definitions of but-for-causes and of $\dhamm$-counterfactual causes have one significant difference:
A but-for-cause $X$ merely requires the existence of values to assign to the variables in $X$ by $\mathsf{intervention}$ such that the effect is avoided. A $\dhamm$-counterfactual cause $C$ in $\cT$  requires that for all possible interventions on the variables in $X$, the effect is avoided. This universal quantification originates from the universal 
quantification over most similar worlds in the Stalnaker-Lewis semantics of counterfactual causality.

%\begin{remark}
The minimality requirement of but-for-causes  does not have a counterpart in the  definition of $d$-counter\-factual causes.
This allows us to assert that a candidate set of states $C$ is a $d$-counterfactual cause for an effect even if it contains redundancies.
When trying to find $d$-counterfactual causes for a given effect, on the other hand, of course trying to find (cardinality-)minimal causes is a reasonable option.
%\end{remark}

Besides but-for causality, 
we can also capture actual causality as in \cite{Halpern15} in our framework in the case of Boolean structural equation models.
This is demonstrated in Appendix \ref{app:actual_causality}.

\section{Counterfactual causality in reachability games}

The counterfactual notion of causality introduced and investigated in the previous section can be applied to reachability games $\game$: 
We take the perspective of a player $\Pi$. 
Given a strategy $\sigma$ for the 
opponent and a play in which $\Pi$ lost, we  apply the definition to the transition system obtained from $\sigma$ and $\game$ and the given play. 
This allows us to analyze whether avoiding a certain set of states while playing against strategy $\sigma$ as similarly as possible to the given play would have allowed $\Pi$ to win.
Depending on whether we take the perspective of $\ReachPl$ or $\SafePl$, the effect that the player loses the game is a safety or reachability property, which we considered as effects in transition systems.
 The need to be given a strategy for the opponent, however, constitutes a major restriction to the usefulness of this approach.

\subsection{$D$-counterfactual causality}
\label{sec:counterfactual_games}

We provide a definition of counterfactual causality in reachability games in the sequel in which we only need the strategy $\sigma$ with which the player $\Pi$  played and are interested in  why the strategy $\sigma$  allows the opponent to win the game.
Since both players have optimal MD-strategies in a reachability game, we restrict ourselves to 
MD-strategies in the definition. 
%Consequently, we use distance functions on MD-strategies. 

\begin{definition}
	\label{def:game_conterfactual-cause}
Let $\game$ be a reachability game with target set $\LocsT$. Let $\Pi$ be one of the two players $\ReachPl$ or $\SafePl$ and let $\sigma$ be a MD-strategy for player $\Pi$. Let $C$ be a set of locations disjoint from $\LocsT$. Let $D$ be a distance function on MD-strategies.
We say that $C$ is a $D$-counterfactual cause for the fact that $\Pi$ loses using $\sigma$  if 
\begin{enumerate}
\item there are $\sigma$-plays that reach $C$ on which $\Pi$ loses,
\item there is an MD-strategy $\tau$ for player $\Pi$ that avoids $C$ (i.e., there is no $\tau$-play reaching $C$),
\item  all MD-strategies $\tau$ for player $\Pi$, that avoid $C$ and that have minimal $D$-distance to $\sigma$ among the strategies avoiding $C$, are winning for $\Pi$.
\end{enumerate} 
\end{definition}

If we take the perspective of player $\Pi$ in game $\game$ where the opponent $\bar\Pi$ does not control any locations, MD-strategies for $\Pi$ satisfying condition 1 of the definition are essentially simple paths satisfying a safety or reachability effect property (with additional information on the states that are not visited by the path). To some extent, the definition can now be seen as a generalization of the definition for transition systems for suitable distance functions $D$:
We say a strategy distance function $D$  \emph{generalizes a path distance function} $d$ if in games where $\bar\Pi$ does not control any location, for all strategies $\sigma, \tau$ for $\Pi$, we have  $D(\sigma,\tau) = d(\pi_{\sigma},\pi_\tau)$ where $\pi_\sigma$ and $\pi_\tau$ are the unique $\sigma$- and $\tau$-plays.
The definition that $C$ is a  $D$-counterfactual cause for $\sigma$ losing the game agrees with the definition that $C$ is a  $d$-counterfactual cause on $\pi_{\sigma}$ for 
$\lozenge \LocsT$ or $\Box \neg \LocsT$ 
 in acyclic games  in this case. In cyclic games, there is one caveat: The definition for games quantifies only over MD-strategies which induce a play that is a simple path or simple lasso. The definition for transition systems quantifies over more complicated paths as well.

\paragraph*{Hausdorff distance $\dpref^H$ based on the prefix metric $\dpref$.}
A way to obtain a strategy distance function generalizing a given path distance function is the use of the Hausdorff distance on the set of plays of the strategies \cite[Section~6.2.2]{DelfourZ-11}:
	Let \winstrategy and \loosestrategy be two MD-strategies, and 
	$d$ be a distance function over plays. 
	The Hausdorff distance $d^H$ based on $d$  is defined by 
\[	{d^H}(\loosestrategy, \winstrategy) = 
	\max \left\{ \sup_{ \text{ $\loosestrategy$-plays }\looseplay} \,
	\inf_{ \text{ $\winstrategy$-plays }\winplay} \,
	d(\looseplay, \winplay),  
	\sup_{ \text{  $\winstrategy$-plays }\winplay} \,
	 \inf_{ \text{  $\loosestrategy$-plays } \looseplay} \,
	d(\looseplay, \winplay) \right\}.
	\]
	Let us consider the Hausdorff distance $\dpref^H$ based on the prefix metric $\dpref$ assuming that all states have a unique label.
For two strategies $\sigma$ and $\tau$ for $\SafePl$, the distance $\dpref^H(\sigma,\tau)$ 
is $2^{-n}$ where $n$ is the least natural number such that there is a prefix of length $n$ of a $\tau$-play that is not a prefix of a $\sigma$-play, or vice versa.
In order to find strategies that are as similar as possible to a given strategy $\sigma$, we hence have to consider strategies that follow $\sigma$ for as many steps as possible.
This leads to an algorithm for checking $\dpref^H$-counterfactual causality in reachability games that shares some similarities with the algorithm for checking $\dpref^{\AP}$-counterfactual causality in transition systems. 
This algorithm is given in Appendix~\ref{app:proof_counterfactual_games}.

\begin{restatable}{theorem}{checkprefgames}
	\label{thm:checking-pref-game}
Let $\game=(V,\locinit,\Trans)$ where $V=\LocsReach\uplus\LocsSafe\uplus\LocsT$ be a reachability game with target set $\LocsT$ and $\sigma$ a MD-strategy for player $\Pi$. Let $C$ be a set of locations disjoint from $\LocsT$.
We can check in polynomial time whether $C$ is a $\dpref^H$-counterfactual cause for the fact that $\Pi$ is losing using $\sigma$  in $\game$.
\end{restatable}

For the Hausdorff lifting of $\dhamm$ or $\dlev$, the resulting notion of counterfactual causes in games is more complicated.
If we try to adapt the approach used in transition systems, we need a way to capture the minimum distance of a given strategy to the closest winning strategies. 
However,  shortest path games (as extension of the weighted transition systems used for $\dhamm$- and $\dlev$-counterfactual causes in transition systems) cannot be employed in an obvious way.
In this paper,  we now instead consider two further distance functions related to the Hamming distance for which we can provide algorithmic results.

\paragraph*{Hamming strategy distance.} 
Let \loosestrategy and \winstrategy be two MD-strategies for \Pl in  \game, we define the Hamming strategy distance  by $\dhamm^s(\loosestrategy, \winstrategy) = 
|\{\loc \in \Locs \mid \loosestrategy(\loc) \neq \winstrategy(\loc)\}|$. As the Hamming distance on paths counts positions at which traces differ, the Hamming strategy distance counts positions at which two MD-strategies differ. 
Using a similar proof using shortest-path games~\cite{KhachiyanBBEGRZ-07} as for Theorem~\ref{thm:checking_dhamm}, 
we obtain the following polynomial-time result in the case of \emph{aperiodic} games. The  proof is given in Appendix~\ref{app:proof_counterfactual_games}.
\begin{restatable}{theorem}{checkhammgames}
	\label{thm:checking-hamm-game}
	Let $\game=(V,\locinit,\Trans)$ where $V=\LocsReach\uplus\LocsSafe\uplus\LocsT$ be an acyclic reachability game with target set $\LocsT$ and $\sigma$ a MD-strategy for player $\Pi$. Let $C$ be a set of locations disjoint from $\LocsT$.
	We can check in polynomial time whether $C$ is a $\dhamm^s$-counterfactual cause for the fact that $\Pi$ is losing using $\sigma$  in $\game$.
\end{restatable}

\paragraph*{Hausdorff-inspired distance \dHH.}
The distance function \dHH computes the number of vertices where two MD-strategies make 
distinct choices along each play of both MD-strategies. It  hence has some similarity to a Hausdorff-lifting of the Hamming distance on paths. This Hausdorff-lifting, however, counts the number of 
\emph{occurrences} of  vertices at which two paths differ (in their label).
Instead, for a play $\winplay=\loc_0\loc_1\dots$ and a strategy \loosestrategy for \Pl, we define 
the \emph{distance between 
	\loosestrategy and \winplay} $\dist(\winplay, \loosestrategy)$ as the number of vertices  
$\loc \in \Locs_{\Pl}$ (i.e., not the number of occurrences) such that there exists $i \in \N$ with $\loc = \loc_i$ in \winplay, and $\loosestrategy(\loc_i) \neq 
(\loc_i, \loc_{i+1})$.We define
\dHH for  two strategies $\winstrategy, \loosestrategy$ by 
\[
\dHH(\winstrategy, \loosestrategy) = \max \big(
\sup_{\winplay \mid \text{\winstrategy-play}} \dist(\winplay, \loosestrategy), 
\sup_{\looseplay \mid \text{\loosestrategy-play}} \dist(\looseplay, \winstrategy) 
\big).
\]
To simplify the notation, we define
$\dstrat{\winstrategy}(\loosestrategy) = 
\sup_{\winplay \mid \text{\winstrategy-play}} \dist(\winplay, \loosestrategy)$.
Some properties of these distances are given in 
Appendix~\ref{app:distance-strategies} and, in particular, we prove that the threshold problem for
\dHH is NP-complete via a reduction from the \emph{longest simple path problem}:

\begin{restatable}{proposition}{propdHH}
	\label{prop:dHH_NP-c}
	Let \game be a reachability game, \loosestrategy, \winstrategy be two  
	MD-strategies for \Pl, and $k \in \N$ be a  threshold. Then deciding 
	if $\dHH(\winstrategy, \loosestrategy) \geq k$ is NP-complete.
\end{restatable}

The  proposition explains why understanding $\dHH$-counterfactual causes is complex. We leave a further investigation of such notions for future work. As a first step toward a better understanding, we turn our attention to 
a conceptually simpler notion, the explanation induced by a counterfactual cause.

\begin{example}
	Let us illustrate counterfactual causes according to distances on strategies. 
	We consider the reachability game depicted in the left of \figurename{~\ref{fig:game_strat-cause-2}} 
	and the non-winning strategy $\loosestrategy$ for \ReachPl depicted in~green. 
	Under $\dpref^H$ or $\dHH$, the counterfactual cause for \ReachPl is $\{v_2, v_3\}$. 
	Indeed, there exists one play that reaches $v_3$ and loses for \ReachPl, 
	and there exists a unique strategy that avoids $\{v_2, v_3\}$ by 
	changing the choice of $\loosestrategy$ in $v_1$. Moreover, this 
	counterfactual cause is minimal since $\{v_3\}$ is not a cause. 
	Indeed, the (losing) strategy that differs from $\loosestrategy$ 
	in v$_0$ and $v_1$ avoids $\{v_3\}$ with a minimal distance to $\loosestrategy$, 
	i.e.\ $2^{-2}$ for $\dpref^H$ and $1$ for $\dHH$. 
	Under $\dhamm^s$, the counterfactual cause for \ReachPl is $\{v_3\}$. 
	Indeed, two strategies exist with a distance of $1$ to
	$\loosestrategy$ according to the vertex where \ReachPl 
	changes its decision. In these two strategies, only one 
	avoids $\{v_3\}$: the strategy where \ReachPl change its decision in $v_1$.  
	\markend
\end{example}

\subsection{$D$-counterfactual explanation}
\label{sec:explain-strat}

Given a $D$-counterfactual cause, we want to explain what is wrong in the losing strategy for \Pl. In particular, we are interested in sets of locations $C$ such that \Pl could have won the game if she had not made the decisions of $\sigma$ in the locations in $C$. 

\begin{definition}
	\label{def:game_quasi-strat-cause}
	Let \game be a reachability game  and \loosestrategy be a non-winning 
	MD-strategy for \Pl. Let $E \subseteq \Locs_{\Pl}$.
	We call $E$ an \emph{explanation} in \game under \loosestrategy if there exists a winning 
	MD-strategy \winstrategy such that for 
	all vertices $\loc \in \Locs_{\Pl}$, $\winstrategy(\loc) = \loosestrategy(\loc)$ 
	iff $\loc \notin E$. We call such a $\tau$ an
	\emph{$E$-distinct \loosestrategy-strategy}.
\end{definition}

%If a quasi-strat-cause $C$ for $\sigma$ is large, a winning $C$-distinct $\sigma$-strategy might be very different from $\sigma$ and hence $C$ might not be a suitable explanation what went wrong with $\sigma$. To capture this concern formally, we define strat-causes by requiring that the distance from $\sigma$  to a winning $C$-distinct $\sigma$-strategy is as small as possible:
%
%
%\begin{definition}
%	\label{def:game_strat-cause}
%	Let \game be a reachability game and \loosestrategy be a non-winning 
%	MD-strategy for \Pl. Let $d$ be a distance function for MD-strategies and 
%	$C \subseteq \Locs_{\Pl}$ be a quasi-strat-cause for \game under \loosestrategy. We 
%	say that $C$ is a \emph{$d$-strat-cause}, if there exists a $C$-distinct 
%	\loosestrategy-strategy  $\winstrategy$ with 
%	$d(\winstrategy, \loosestrategy) = 
%	\min \{d(\strategy, \loosestrategy)\mid {\text{$\strategy$ is a winning MD-strategy for $\Pl$}}\}$.
%	\markend
%\end{definition}

We note that the definition of an explanation does not refer to a distance function. 
However, given a $D$-counterfactual cause, we can  compute an explanation no matter which  distance $D$ is used. 

\begin{proposition}\label{prop:difference_explanation}
	Let $\game=(V,\locinit,\Trans)$ be a reachability game, $D$ a distance function on strategies and \loosestrategy be a non-winning 
	MD-strategy for \Pl. Let $C \subseteq \Locs$ be a $D$-counterfactual cause.
	We can compute an explanation $E$ (from $C$) in polynomial time.
\end{proposition}
\begin{proof}
	Let $\game'=(V\setminus C,\locinit,\Trans)$ be the reachability game. Since $D$ is a $D$-counterfactual cause, we 
	know that there exists a winning strategy $\tau$ in $\game'$. We can compute this strategy in  time polynomial in the size of $\game'$ with the attractor method and we define $E = \{v \mid \loosestrategy(v) \neq \tau(v) \}$.
\end{proof}

A winning strategy differing from $\sigma$ in $E$ might not have much in common with $\sigma$. 
For this reason, explanations that point out changes in the
decisions of $\sigma$ in $E$ that enforce only the minimal necessary change to obtain a winning strategy $\tau$ from $\sigma$ are of particular interest. 
We can use a distance function $D$ to quantify how much a strategy needs to be changed.

\begin{definition}
	Let \game be a reachability game  and \loosestrategy be a non-winning 
	MD-strategy for \Pl. For a distance function $D$ for MD-strategies, we call a explanation 
	$E$ a \emph{$D$-minimal explanation}, if there exists a winning $E$-distinct 
	\loosestrategy-strategy  $\winstrategy$ with 
	$d(\winstrategy, \loosestrategy) = 
	\min \{d(\strategy, \loosestrategy)\mid {\text{$\strategy$ is a winning MD-strategy for $\Pl$}}\}$.
\end{definition}

For a strategy $\sigma$ and an explanation $E$, the distance $\dhamm^s(\sigma,\tau)$ for an $E$-distinct $\sigma$-strategy $\tau$ is precisely $|E|$. So, $\dhamm^s$-minimal explanations are cardinality-minimal explanations.

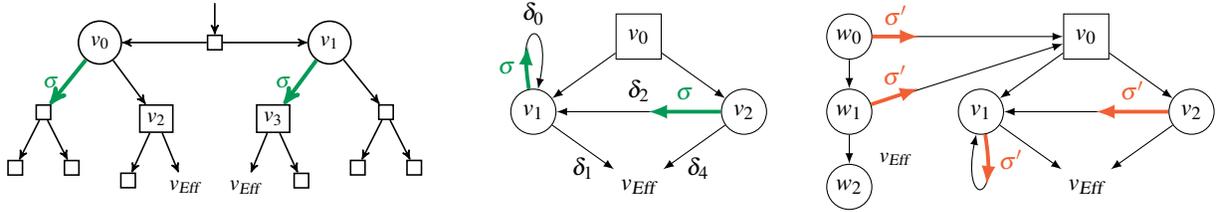
\begin{figure}[t]
\centering
	\begin{minipage}{.35\textwidth}
		\scalebox{0.75}{\begin{tikzpicture}
			[scale=1,->,>=stealth',auto ,node distance=0.5cm, thick]
			\tikzstyle{r}=[thin,draw=black,rectangle]
			
			\node[scale=1,rectangle, draw] (start) {};
			\draw[<-] (start) --++(0,+0.7);
			\node[scale=1, circle, draw, left=1.5 of start] (step0) {$v_0$};
			\node[scale=1, circle, draw, right=1.5 of start] (step1) {$v_1$};
			
			\node[scale=1, rectangle, draw, below=.7 of step0,xshift=-1cm] (step00) {};
			\node[scale=1, rectangle, draw, below=.7 of step0,xshift=1cm] (step01) {$v_2$};
			
			\node[scale=1, rectangle, draw, below=.7 of step00,xshift=-.5cm] (step000) {};
			\node[scale=1, rectangle, draw, below=.7 of step00,xshift=.5cm] (step001) {};
			
			\node[scale=1, rectangle, draw, below=.7 of step01,xshift=-.5cm] (step010) {};
			\node[scale=1, rectangle, draw=none, below=.7 of step01,xshift=.5cm] (step011) {$\locT$};

			\node[scale=1, rectangle, draw, below=.7 of step1,xshift=-1cm] (step10) {$v_3$};
			\node[scale=1, rectangle, draw, below=.7 of step1,xshift=1cm] (step11) {};

			\node[scale=1, rectangle, draw=none, below=.7 of step10,xshift=-.5cm] (step100) {$\locT$};
			\node[scale=1, rectangle, draw, below=.7 of step10,xshift=.5cm] (step101) {};
			
			\node[scale=1, rectangle, draw, below=.7 of step11,xshift=-.5cm] (step110) {};
			\node[scale=1, rectangle, draw, below=.7 of step11,xshift=.5cm] (step111) {};

			\draw[color=black,->] (start) edge  (step0) ;
			\draw[color=black,->] (start) edge  (step1) ;
			\draw[color=black,->, ForestGreen,line width=2pt] (step0) edge node[left] {\textcolor{ForestGreen}{$\loosestrategy$}} (step00) ;
			\draw[color=black,->] (step0) edge  (step01) ;
			\draw[color=black,->, ForestGreen,line width=2pt] (step1) edge node[left] {\textcolor{ForestGreen}{$\loosestrategy$}} (step10) ;
			\draw[color=black,->] (step1) edge  (step11) ;
			\draw[color=black,->] (step00) edge  (step001) ;
			\draw[color=black,->] (step00) edge  (step000) ;
			
			\draw[color=black,->] (step01) edge  (step011) ;
			\draw[color=black,->] (step01) edge  (step010) ;
			
			\draw[color=black,->] (step10) edge  (step101) ;
			\draw[color=black,->] (step10) edge  (step100) ;
			
			\draw[color=black,->] (step11) edge  (step111) ;
			\draw[color=black,->] (step11) edge  (step110) ;

			\end{tikzpicture}}
	\end{minipage}
	\hfill
	\begin{minipage}{.6\textwidth}
		\begin{tikzpicture}[xscale=.7,every node/.style={font=\footnotesize}, 
		every label/.style={font=\scriptsize}]
		\node[PlayerSafe] at (0,0) (s0) {$\loc_0$};
		\node[PlayerReach] at (-2,-1) (s1) {$\loc_1$};
		\node[PlayerReach] at (2, -1) (s2) {$\loc_2$};
		\node[target] at (0,-2) (e) {$\locT$};
		\node[strat] at (-2.1, 0) (st1) {};
		\node[strat] at (0, -1) (st2) {};
		
		\draw[->]
		(s0) edge (s1)
		(s0) edge (s2)
		(s1) edge[loop above,min distance=10mm] node[above] {$\trans_0$} (s1)
		(s1) edge node[below] {$\trans_1$} (e)
		(s2) edge node[above] {$\trans_2$} (s1)
		(s2) edge node[below, xshift=.05cm] {$\trans_4$} (e)
		;
		
		\draw[ForestGreen, ->, line width=0.5mm]
		(s1) edge[bend left=10] node[left] {\textcolor{ForestGreen}{$\loosestrategy$}} (st1)
		(s2) edge node[above] {\textcolor{ForestGreen}{$\loosestrategy$}}(st2)
		;
		
		\begin{scope}[xshift=8.5cm]
		\node[PlayerSafe] at (0,0) (s0) {$\loc_0$};
		\node[PlayerReach] at (-2,-1) (s1) {$\loc_1$};
		\node[PlayerReach] at (2, -1) (s2) {$\loc_2$};
		\node[target] at (0,-2) (e) {$\locT$};
		
		\node[strat] at (-1.9, -2.1) (st1) {};
		\node[strat] at (0, -1) (st2) {};
		\node[strat] at (-3, 0) (st3) {};
		\node[strat] at (-3, -.65) (st4) {};
		
		\node[PlayerReach] at (-4.5,0) (w0) {$w_0$};
		\node[PlayerReach] at (-4.5, -1) (w1) {$w_1$};
		\node[PlayerReach,label={30:$\locT$}] at (-4.5, -2) (w2) {$w_2$};
		
		\draw[->]
		(s0) edge (s1)
		(s0) edge (s2)
		(s1) edge[loop below,min distance=10mm] (s1)
		(s1) edge (e)
		(s2) edge (s1)
		(s2) edge  (e)
		(w0) edge (w1)	
		(w0) edge (s0)
		(w1) edge (w2)	
		(w1) edge (s0)	
		;
		
		\draw[RedOrange, ->, line width=0.5mm]
		(s1) edge[bend left=10] node[right] {\textcolor{RedOrange}{$\loosestrategy'$}} (st1)
		(s2) edge node[above] {\textcolor{RedOrange}{$\loosestrategy'$}}(st2)
		(w0) edge node[above] {\textcolor{RedOrange}{$\loosestrategy'$}} (st3)
		(w1) edge node[above] {\textcolor{RedOrange}{$\loosestrategy'$}}(st4)
		;
		\end{scope}
		
		\end{tikzpicture}
	\end{minipage}	
	\caption{On the left and in the middle, two reachability games with initial vertex $v_0$ and 
		strategy $\sigma$ for $\ReachPl$ (depicted in green). On the right, the reachability game 
		obtained by reduction of Corollary~\ref{cor:check-hamm-game} from the game depicted in the middle 
		with initial vertex $w_0$ and $\loosestrategy'$ be a non winning strategy for \ReachPl.}
	\label{fig:game_strat-cause-2}
\end{figure}

\begin{example} \label{ex:quasi-strat-cause}
	Let us illustrate explanations and $D$-minimal explanations.	
	We consider the reachability game \game where \ReachPl wins depicted 
	in the left of Figure {\ref{fig:game_strat-cause-2}} with \loosestrategy, a non-winning 
	MD-strategy for \ReachPl, depicted in green. 
	We note that $E = \{\loc_1, \loc_2\}$ 
	is an explanation in \game under \loosestrategy. A  winning 
	$E$-distinct \loosestrategy-strategy \winstrategy for \ReachPl is given by
	$\winstrategy(\loc_1) = \trans_1$ and $\winstrategy(\loc_2) = \trans_4$. 
	However, $E$ is not a $\dHH$-minimal explanation or $\dhamm^s$-minimal explanation. 
	Clearly, $\dhamm^s(\winstrategy,\loosestrategy)=2$. Further, also $\dHH(\winstrategy,\loosestrategy)=2$ as the $\sigma$-play 
	$\loc_0\loc_2\loc_1^{\omega}$ visits two states, namely $\loc_2$ and $\loc_1$ at which $\sigma$ and $\tau$ make different decisions.
	The set $E^\prime=\{\loc_1\}$, however is a $\dHH$-minimal explanation and $\dhamm^s$-minimal explanation:
	the $E^\prime$-distinct $\sigma$-strategy $\tau^\prime$ choosing $\trans_1$ in $\loc_1$ and behaving like $\sigma$ in $\loc_2$ wins and has $\dhamm^s$- and $\dHH$-distance $1$ to $\sigma$.
	As any winning strategy has at least distance $1$ to $\sigma$, $E^\prime$ is hence a $D$-minimal explanation for both distance functions.
	\markend
\end{example}

%\subsubsection{Checking $D$-minimality of explanations}
\label{subsec:strat_find-strat}

For $D$-minimal explanations, it is central to find a winning MD-strategy that minimises 
the distance $D$ to the given losing strategy $\sigma$. We take a look at this problem from the point of view of 
\ReachPl and prove that for $\dhamm^s$ and $\dHH$ the associated threshold problems are not in P if P$
\not=$NP.

%To look after the hardness of the problem, we use its decision version: letting \game 
%where \ReachPl wins, \loosestrategy be a non-winning MD-strategy for \ReachPl and 
%$k \in \N$, the \emph{finding a $k,d$-closed winning MD-strategy for \ReachPl} 
%problem that asks if there exists a winning MD-strategy \winstrategy for \ReachPl 
%such that $d(\winstrategy, \loosestrategy) \leq k$. In particular, we prove that the 
%finding a $k,d$-closed winning MD-strategy problem for \ReachPl is NP-hard when 
%we consider Hamming strategy distance or \dHH distance.

\begin{restatable}{theorem}{stratFindingNPc}
	\label{thm:strat_finding-NP-c}
	Given a game $\game$,  a losing strategy $\sigma$ for $\ReachPl$, and $k\in \mathbb{N}$,
	deciding if there exists a winning MD-strategy \winstrategy for \ReachPl 
such that $\dhamm^s(\winstrategy, \loosestrategy) \leq k$ is  NP-complete. Further, the problem whether there is a winning MD-strategy \winstrategy with
$\dHH(\winstrategy, \loosestrategy) \leq k$ is not in P if P$\not=$NP.
\end{restatable}
\begin{proof}[Proof sketch]
	To establish the NP upper bound for $\dhamm^s$, we can guess a MD-strategy $\tau$ for \ReachPl and check in polynomial time whether 
	it is winning and whether $\dhamm^s(\winstrategy, \loosestrategy) \leq k$.
	For the NP-hardness for $\dhamm^s$, we provide a polynomial-time many-one reduction from the  NP-complete
	decision version of the \emph{feedback vertex set} 
	\cite{Karp1972}. Given a cyclic (directed) graph $G$, this problem asks whether there is a set $S$ of size at most $k$
	such that if we remove this set, $G \setminus S$ becomes acyclic. 
	For the problem for \dHH, we provide a polynomial-time Turing reduction from the same problem.
	 A detailed proof  is 
	given in Appendix~\ref{app_FVS-CSS}. 
%	
%	Intuitively, from an instance of feedback vertex set problem $G = (L, \Delta)$, we 
%	define $\game = (\Locs, \locinit, \Trans)$ by letting $L$ with transitions $\Delta$ be controlled by $\SafePl$,
%	but adding a vertex of \ReachPl before each 
%	vertex in $L$ from which a target state is reachable. Thus, vertices of \SafePl keep the possible choices 
%	from $G$, and vertices of \ReachPl can choose between staying in the game  or reaching the effect (see Figure~\ref{fig:ex_feedback}).
%	Furthermore, an initial vertex with edges to all states in $L$ controlled by $\SafePl$ is added and $\sigma$ is the losing strategy for $\ReachPl$ that always chooses to stay in the game.
%	A detailed proof  is 
%	given in Appendix~\ref{app_FVS-CSS}.  
\end{proof}

From these  results, we deduce that also checking if a given explanation is 
$D$-minimal cannot be done in polynomial time  if P$\not=$NP  (see Appendix~\ref{app_FVS-CSS} for the proof).

\begin{restatable}{corollary}{checkdhammsstrat}
	\label{cor:check-hamm-game}
	Let \game be a reachability game, \loosestrategy be a 
	non-winning MD-strategy for \ReachPl,  and $E \subseteq \Locs$. The problem to check if $E$ is a $\dhamm^s$-minimal explanation in $\game$ for \loosestrategy
	is  coNP-complete. The problem to check if $E$ is a $\dHH$-minimal explanation in $\game$ for \loosestrategy
	is not in P if P$\not=$NP.
\end{restatable}

Despite the hardness in the general case, if $\game^{\loosestrategy}$ is acyclic, we prove 
that we can compute the winning MD-strategy that minimises the $\dHH$-distance to $\sigma$  in polynomial time.
 From this strategy, a $\dHH$-minimal explanation can then be computed as in Proposition \ref{prop:difference_explanation}.
 The  proof of 
	Theorem~\ref{thm:find-trans_poly}  (in Appendix~\ref{app:poly-algo})  constructs a shortest-path game~\cite{KhachiyanBBEGRZ-07} without negative weights in which an optimal strategy, that leads to the desired winning strategy in the original game, can be computed in polynomial time.

\begin{restatable}{theorem}{findTransPoly}
	\label{thm:find-trans_poly}
	Let \game be reachability game where \ReachPl wins, and \loosestrategy be a 
	non-winning MD-strategy for \ReachPl such that $\game^{\loosestrategy}$ is 
	acyclic. Then, we can compute a winning MD-strategy $\winstrategy$ that minimizes the distance \dHH to 
	\loosestrategy in polynomial time.
\end{restatable}
\section{Conclusion and Outlook}
\label{sec:conclusion}
The introduced notion of $d$-counterfactual cause for a distance function $d$ in transition systems turned out to be checkable in polynomial time for the distance functions $\dpref$, $\dhamm$, and $\dlev$ and it, hence, has the potential to be employed in efficient tools to provide understandable explanations of the behavior of a system. 
In our results for safety effects $\Phi$, one caveat remains: we  only considered finite executions reaching a cause candidate $C$ and satisfying $\Phi$ in the input to the checking algorithms. Allowing also finitely representable, e.g., ultimately periodic paths, constitutes a natural possible  extension, which requires adjustments in the provided algorithms.
 
The problem of finding good causes remains as future work: Whenever causality can be checked in polynomial time,  there is an obvious non-deterministic polynomial-time upper bound on the problem to decide whether there are causes below a given size, but the precise complexities are unclear. A further idea is to use the distance function to assess how good a cause is by considering  the distance from the actual execution to the closest executions avoiding a cause. For reachability effects and the prefix and Hamming distance, the set of direct predecessors optimizes this distance. For other distance functions or safety causes, this measure could, nevertheless, be more useful.
The search for similar  measures for the quality of causes constitutes an interesting direction for future  work.

In reachability games, we saw that the analogous definition of $D$-counterfactual causes can be checked in polynomial time for the Hausdorff-lifting $\dpref^H$ of the prefix metric, as well.
For other distance functions, 
the definition seems to lead to complicated notions due to the involved quantification over all MD-strategies avoiding the cause and having a minimal distance to a given strategy.
A closer investigation of these notions might, nevertheless, be a fruitful subject for future research.
However, our analysis of the 
 conceptually simpler $D$-minimal explanations provides insights into the complications one might encounter here.
 For the Hausdorff-inspired distance function $\dHH$, we showed that already the threshold problem for the distance between two given MD-strategies is NP-hard.
 Furthermore, for the relatively simple distance function $\dhamm^s$, checking the $\dhamm^s$-minimality of an explanation is in coNP-complete. For the Hausdorff-inspired distance function \dHH, checking \dHH-minimality is
  not in P unless P$=$NP.

\bibliographystyle{abbrv}

\bibliography{main}

\newpage
\appendix
\section{Omitted proofs of Section \ref{sec:counterfactual_in_TS}}
\label{app:counterfactual_in_TS}

\subsection{Omitted proofs of Section \ref{sub:check_TS}}

\checkprefAPTS*

\begin{proof}
We have to find the last index $i<n$ such that $C$ is avoidable after having seen the trace $L(s_0\dots s_i)$.
So, we iteratively construct the set of states that are reachable via a path $\zeta$ whose trace is a prefix of the trace of $\pi$.
We set $T_0=\{\sinit\}$. Now, for each $i<n$, we construct the set
\[
T_{i+1} = \{t\in S \mid L(t)=L(s_{i+1})\text{, $t\not \in C$, and there is an $s\in T_i$ with $(s,t )\in \to$ }\}.
\]
For each $i<n$, the set $T_i$ contains all states that are reachable via a path $\zeta$ with trace $L(s_0\dots s_i)$ that does not reach $C$. Furthermore, the sequence of sets $T_i$ with $i<k$ can be constructed in polynomial time.

Now, we check for each $i<n$, whether there is a state $t\in T_i$ with $t\vDash \exists \Box \neg C$, i.e., whether $C$ is still avoidable from some state in $T_i$. Let $j+1$ be the least index for which this is not  the case. Note that if the least such index is $0$, then $C$ is not avoidable and hence not a cause.
So, the $\dpref^{\AP}$ most similar paths to $\pi$ that avoid $C$ are those that share the prefix $L(s_0\dots s_j)$ of the trace with $\pi$ and we know that there are such paths.

Now, we can proceed by checking for all states $t\in T_j$ whether they satisfy $t\vDash \forall (\Phi \to \lozenge C)$ in polynomial time as described below. If this is the case, all $\dpref^{\AP}$-closest paths to $\pi$ that avoid $C$ do not satisfy $\Phi$ and hence $C$ is a $\dpref^{\AP}$-counterfactual cause for $\Phi$ on $\pi$. Otherwise, there is a path $\pi^\prime$ sharing the prefix  $L(s_0\dots s_j)$ of the trace with $\pi$, avoiding $C$, and satisfying $\Phi$. As $\pi^\prime$ belongs to the paths avoiding $C$ with minimal $\dpref^{\AP}$-distance to $\pi$, the set $C$ is not a $\dpref^{\AP}$-counterfactual cause for $\Phi$ on $\pi$ in this case.

To check whether a state $t$ satisfies $t\vDash \forall (\Phi \to \lozenge C)$, we distinguish whether $\Phi$ is a reachability or safety property:
 If $\Phi=\lozenge E$, we observe that checking whether 
$t \vDash \forall (\Phi \to \lozenge C)$ can be done by removing all states in $C$ from $\cT$ and afterwards checking whether $E$ is reachable from $t$ in the resulting transition system which can be done in polynomial time. 
If $\Phi=\Box \neg E$, we make the states in $C$ terminal in $\cT$ and check in polynomial time whether there is a path from $t$ not reaching $E$ or $C$. If such a path exists, $t \not \vDash \forall (\Phi \to \lozenge C)$. Otherwise, $t  \vDash \forall (\Phi \to \lozenge C)$.

As the construction of the sets $T_i$ and the checks whether there is a state $t\in T_i$ with $t\vDash \exists \Box \neg C$ and 
whether all states $t\in T_i$ satisfy $t\vDash \forall (\Phi \to \lozenge C)$ can be done in polynomial time, 
the problem is indeed decidable in polynomial time. For a more efficient computation, 
one could  check whether there is a state in  $T_i$ that satisfies  $\exists \Box \neg C$ before constructing the set $T_{i+1}$. 
\end{proof}

\checkhammTS*

\begin{proof}
We provide the proof for the case that $\Phi=\lozenge E$. As the terminal states $E$ are located at the last layer and $\cT$ is acyclic, the property $\Box \neg E$ can be expressed as a reachability property with a target set containing all terminal states not in $E$ as well and the proof hence works analogously.

Recall that given a path $\rho=t_1 \dots  t_k$, the distance 
\[
\dhamm (\pi,\rho) = |\{1\leq i \leq k \mid L(s_i)\not = L(t_i)\}|.
\]
 We will equip the states in $S$ with a weight function $\wgt\colon S\to \{0,1\}$ such that the distance of a path to $\pi$ is equal to the accumulated weight of that path. For this purpose, we annotate states with the layer they are on: A state $t$ on layer $i$ is denoted by $(t,i)$.
We define
\[
\wgt((t,i))=\begin{cases}
1 & \text{if $L(t)\not=L(s_i)$,} \\
0 & \text{if $L(t)=L(s_i)$.}
\end{cases}
\]
Now, any path $\rho=t_1\dots t_k$ in $\cT$ that ends on the last layer $k$ satisfies 
\[
\wgt(\rho)\eqdef \sum_{i=0}^k \wgt(t_k) = \dhamm(\pi,\rho).
\]
When checking whether $C$ is a $\dhamm$-counterfactual cause, we have to look at paths not reaching $C$. Hence, we remove all states in $C$ from $\cT$ to obtain the weighted transition system $\cT_C$. The question is whether, among the path of $\cT_C$ that are most similar to the original execution $\pi$  according to $\dhamm$, there is a path that reaches $E$.
This question is now easy to answer: We compute the shortest path from $\sinit$ to the last layer $k$ in $\cT_C$ and denote the weight by $W_k$. If there is no path to layer $k$ anymore, then $C$ is unavoidable in $\cT$ and not a cause. 
Afterwards, we compute the shortest path from $\sinit$ to $E$ in $\cT_C$ and call the weight $W_E$. If such a path does not exist, but a path to the $k$th layer exists, then $C$ is a cause because all executions avoiding $C$ also avoid $E$. Otherwise, we compare the two weights: If $W_E=W_k$, then $C$ is not a cause. A shortest path to $E$ in $\cT_C$ is at least as similar to $\pi$ as any other path in $\cT_C$ reaching the last layer. 
If $W_E>W_k$, then the shortest paths in $\cT_C$ reaching the last layer are the paths in $\cT$ that are most similar to $\pi$ among all paths avoiding $C$ and none of these paths reaches $E$. So, in this case $C$ is a $\dhamm$-counterfactual cause for $E$ on $\pi$.

As shortest paths in a weighted transition system can be computed in polynomial time and the transformations made to obtain $\cT_C$ are also easily doable in polynomial time, the described procedure can solve the problem in polynomial time.
\end{proof}

\checkghammTS*

\begin{proof}
We first consider the case that $\Phi=\lozenge E$.
Take $|\pi|$-many copies of the state space $S$ to obtain $S\times\{1,\dots, n\}$.
For each transition $s\to t$ in $\cT$ and $i<n$, there is a transition from $(s,i)\to (t,i{+}1)$. These transitions have weight $0$ if $L(t)=L(s_{i{+}1})$, otherwise, they have weight $1$.
 Furthermore, there are   transitions from $(s,n)\to (t,n)$ with weight $1$.
Finally, from each terminal state $s$ and $i<n$, there is a transition from $(s,i)\to (s,n)$ with weight $n-i$.
We denote the resulting transition system by $\cT_{\pi}$.
Now, any path to a terminal state in this new transition system $\cT_{\pi}$, which are all contained in $S\times \{n\}$, corresponds to a path $\zeta$ in $\cT$ ending in a terminal state
and the weight of the path $\zeta$ is $\dghamm(\pi,\zeta)$.

Hence, we can now check whether $C$ is a $\dghamm$-counter\-factual cause for $E$ on $\pi$ by two shortest path queries as in the proof of Theorem \ref{thm:checking_dhamm}:
We check whether the shortest path in  $\cT_{\pi}$ after removing $C$-states to a terminal state is shorter than the shortest path to a state in $E\times \{n\}$. If this is the case,  $C$ is a $\dghamm$-counter\-factual cause for $E$ on $\pi$; otherwise, it is not.

For the case that $\Phi=\lozenge \neg E$, we use the same construction of $\cT_{\pi}$. To check whether the closest paths to $\pi$ not reaching $C$ do also not reach $E$, we remove the states from $C$ and find a shortest path $\zeta$ to any terminal state if such a path exists. (If no such path exists, the closest paths to $\pi$ avoiding $C$ are infinite. Then, $C$ is a cause only if all paths in 
$\cT_{\pi}$ without $C$ reach $E$.) Afterwards, we compute the shortest path to a state in $E$. $C$ is a cause iff
 the shortest path to $E$ is strictly shorter than $\zeta$.
\end{proof}

\checklevTS*

\begin{proof}
Let $\cT_{\dlev}^\pi$ be as above.
In oder to check whether $C$ is a $\dlev$-counterfactual cause for $\Phi$ on $\pi$, we first remove all states in $C\times\{1,\dots,n\}$ from $\cT_{\dlev}^\pi$.
If there is no path not reaching the removed states anymore, then there are no paths in $\cT$ avoiding $C$ and $C$ is not a cause.

Otherwise, we first consider the case that $\Phi=\lozenge E$ and we compute the shortest path according to the weight function $\wgt$ in $\cT_{\dlev}^\pi$ without the states $C\times\{1,\dots,n\}$  and the shortest path to $E\times\{n\}$. If there is no path to $E\times\{n\}$ anymore, but there are still paths, then $C$ is a $\dlev$-counterfactual cause.
If $E\times\{n\}$ is still reachable,  $C$ is a $\dlev$-counterfactual cause for $\lozenge E$ on $\pi$ if and only if the shortest path to $E\times \{n\}$ is longer than the shortest path to any terminal state in $\cT_{\dlev}^\pi$. Note that infinite paths always have infinite distance to $\pi$ and are hence less similar than any path to $E\times \{n\}$.

To see this, assume that the shortest path $\tau$ to $E\times \{n\}$ has weight $w$ and is at least as short as any path to a terminal state.
This path $\tau$ induces a maximal path in $\cT$ that avoids $C$, reaches $E$, and has $\dlev$-distance $w$ to $\pi$. If there would be another maximal path $\rho$ in $\cT$ that avoids $C$ and has a smaller $\dlev$-distance to $\pi$, this would induce a path $\tau^\prime$ in $\cT_{\dlev}^\pi$ without the states $C\times\{1,\dots,n\}$ to a terminal stste with weight less than $w$.

Analogously, if the shortest path $\tau$ in $\cT_{\dlev}^\pi$ without the states $C\times\{1,\dots,n\}$ to a terminal state is shorter than the shortest path to $E\times \{n\}$, then
this path corresponds to a $\dlev$-closest path to $\pi$ avoiding $C$ and any path avoiding $C$ and reaching $E$ has a larger $\dlev$-distance to $\pi$.

If $\Phi=\Box \neg E$, we can proceed analogously. We compute the shortest path to $E\times \{n\}$ after removing $C\times\{1,\dots,n\}$. If no such path exists, but there are still paths in $\cT_{\dlev}^\pi$ not reaching  $C\times\{1,\dots,n\}$, then $C$ is not a cause as all executions avoiding $C$ also satisfy $\Box\neg E$.

If a path to $E\times \{n\}$ still exists, its length is finite.
Furthermore, we compute the shortest path to a terminal state not in $E\times \{n\}$  in $\cT_{\dlev}^\pi$ without  $C\times\{1,\dots,n\}$. If such a path does not exist, all paths avoiding
$C$ and $E$ have infinite distance to $\pi$ and, hence, $C$ is a $\dlev$-counterfactual cause for $\Phi=\Box\neg E$.
If such a path exists and the shortest such path is at most as long as  the shortest path to 
$E\times \{n\}$, $C$ is not a $\dlev$-counterfactual cause for $\Phi=\Box\neg E$, as there are paths reaching $E$ among the $\dlev$-closest paths to $\pi$. Otherwise, the $\dlev$-closest paths avoiding $C$ reach $E$ and $C$ is  a $\dlev$-counterfactual cause for $\Phi=\Box\neg E$.
\end{proof}

\subsection{Omitted proofs of Section \ref{sec:HP-causality}}

\butforcausality*

\begin{proof}
Let $k=|X|$. If $\xi$ is a path with $\dhamm(\xi,\pi)<k$, then there must be at least one variable $x_i\in X$ where $\xi$ makes a default decision. Hence, $\xi\vDash \lozenge C_X$ in this case.

Next, there is a path $\hat{\pi}$ that departs from $\pi$ exactly at the nodes where decisions for the variables in $X$ are made. This path satisfies $\dhamm(\hat{\pi},\pi)=k$ and $\hat{\pi}\vDash \neg \lozenge C_X$. As all variables are binary, the path $\hat{\pi}$ always assigns the value that the only intervention would assign to the variables in $X$. As $X$ is a but-for-cause, the leave in $\cT$ that is reached by $\hat{\pi}$ does not satisfy the effect.

Finally, we observe that the path $\hat{\pi}$ is the only path with $\dhamm(\hat{\pi},\pi)=k$ and $\hat{\pi}\vDash \neg \lozenge C_X$. Any path $\xi$ that satisfies $\xi\vDash \neg \lozenge C_X$ has to reach one state labelled with $\{\mathsf{intervention}\}$ per variable in $X$. On the other hand, if $\dhamm(\xi,\pi)=k$, the path $\xi$ can only choose the action $\mathsf{intervention}$ $k$ times. If it does so precisely for the variables in $X$, then $\xi=\hat{\pi}$; otherwise $\xi\vDash \lozenge C_X$ as we have seen.
\end{proof}

\subsection{Relation to actual causality}
\label{app:actual_causality}

\paragraph*{Actual causes \`a la Halpern-Pearl}
Let $f_1,\dots,f_n$ be equations of a structural equation system on variables $X_1,\dots, X_n$  and let $E$ be an effect set of variable valuations.
A set $X\subseteq \{X_1,\dots, X_n\}$ of variables is called an \emph{actual cause} for $E$ if 
it is a minimal subset of variables with the following property:
there are values $\alpha_x$ for $x\in X$ and a set $Y\subseteq \{X_1,\dots,X_n\}\setminus X$ such that after obtaining a variable assignment for $X_1,\dots, X_n$ by evaluating the equations in the structural equation system while setting variables $x\in X$ to $\alpha_x$ by intervention and afterwards replacing the values of variables in $Y$ with their values in the default evaluation without intervention, the effect does not occur.

More precisely, let $v_1,\dots,v_n$ be the values of $X_1,\dots,X_n$ in the default evaluation without interventions, i.e.
\[
v_i=f_i(v_1,\dots,v_{i-1})
\]
for all $i$.
Now, let $w_1,\dots, w_n$ be the values of the variables $X_1,\dots, X_n$ if variables $x\in X$ are set to $\alpha_x$ by intervention, while variables $x\not\in X$ are evaluated by the structural equations. So,
\[
w_i =\begin{cases}
 f_{i}(w_1,\dots,w_{i-1})  & \text{ if }X_i\not\in X,\\
 \alpha_{X_i}& \text{ if }X_i\in X,
\end{cases}
\]
for all $i$.
Now, consider the valuation $V$ given by
\[
X_i = \begin{cases}
w_i & \text{ if $x_i\not \in Y$,}\\
v_i & \text{ if $x_i\in Y$.}
\end{cases}
\]
Now, $X$ is an actual cause for $E$ if it is minimal among all sets for which intervention values $\alpha_x$ for $x\in X$ and a set $Y$ exist such that the valuation $V$ as constructed above does not belong to $E$.

\paragraph*{Capturing actual causality}
We now illustrate how actual causality can be captured by counterfactual causality using a weighted Hamming distance as described in Remark \ref{rem:weightedHamming}.

To this end, we extend the transition system $\cT$ constructed from a structural equation model in Section \ref{sec:HP-causality}:
After each leaf $s$ of $\cT$, we enter a further binary tree in which all variables, one after the other, can be reset to their original value that they would have had if no interventions had taken place (i.e., variable $X_i$ may be reset to $v_i$ in the notation of the paragraph above).
If a state is reached by a reset, we label it with $\{\mathsf{reset}\}$.
Finally, we assign the following weights to changes in the labels:
A change from $\emptyset$ to $\{\mathsf{intervention}\}$ has weight $+1$; a change from $\emptyset$ to $\{\mathsf{reset}\}$ has weight $+1/(n+1)$.

The default path $\pi$ is now the path that always uses the  $\mathsf{default}$-action and afterwards does not reset any variables.

Given an actual cause $X$ and a reset set $Y_X$ of minimal cardinality, define $C_X$ as the set of states are reached by using the $\mathsf{default}$-action to assign a value to a variable in $X$ or by choosing not to reset a variable in $Y_X$.

\begin{proposition}
Assume we start with a structural equation system over Boolean variables $X_1,\dots, X_n$ and let all notation be as above. 
Then, $C_X$ is a $\dhamm^w$-counterfactual cause for $E$ on $\pi$ where $\dhamm^w$ is the weighted Hamming distance.
\end{proposition}

\begin{proof}
The path $\zeta$ that makes interventions for all variables in $X$ and that resets all variables in $Y_X$ satisfies
\[
\dhamm^{w}(\zeta,\pi)=|X|+\frac{|Y_X|}{n+1}.
\]
Furthermore, this path $\zeta$ does not reach $C_X$.
As $X$ is an actual cause for $E$ and $Y_X$ is a possible set of variables to reset to avoid $E$ after performing interventions on the variables in $X$, 
the path $\zeta$ ends in a leaf containing a valuation for $X_1,\dots ,X_n$ that does not exhibit the effect $E$.

Reasoning as in the proof of Proposition \ref{prop:butfor}, we see that any path $\xi$ with $\xi\vDash \neg \lozenge C_X$ satisfies
\[
\dhamm^{w}(\xi,\pi)\geq|X|+\frac{|Y_X|}{n+1}
\]
As such a path has to move to states labeled with $\{\mathsf{intervention}\}$ at least whenever a variable in $X$ is assigned a value, i.e., at least $|X|$-many times, and similarly, has to reset at least the variables in $Y_X$, i.e., move to a state labeled  with $\{\mathsf{reset}\}$ at least $|Y_X|$-many times.

Furthermore, any path $\xi^\prime$ that satisfies $\neg \lozenge C_X$ and that is not identical to $\zeta$ has to perform at least one additional intervention or reset and hence satisfies
\[
\dhamm^{w}(\xi^\prime,\pi)>|X|+\frac{|Y_X|}{n+1}.
\]
Note that here we use that all variables are Boolean and hence there is only one possible intervention for each variable.

So, the only path $\zeta$ that has minimal $\dhamm^w$-distance to $\pi$ among the paths avoiding $C_X$ does not reach $E$ and so $C_X$ is a $\dhamm^w$-counterfactual cause for $E$ on $\pi$.
\end{proof}

\section{Omitted proofs of Section \ref{sec:counterfactual_games}}
\label{app:proof_counterfactual_games}
 
In the following proof the notion of an attractor is used. Intuitively, an attractor for a set of states $C$ and a player $\Pl$ is a region of a game in which $\Pl$ can force the game to reach $C$ together with the edges in this region.
Formally, we define an attractor for $\Pl$ to a set of vertices $C$ such 
that $\Locs^0_C = C$, $\Trans^0_C = \emptyset$, and for all $n \in \N$, we have 
\begin{align*}
\Locs^{n+1} = \Locs^n 
&\cup \{\loc \in \Locs_{\Pl} ~|~ \exists (\loc, \loc') \in \Trans 
\text{ such that } \loc' \in \Locs^n \} \\
&\cup \{\loc \in \Locs\setminus\Locs_{\Pl}~|~ \forall (\loc, \loc') \in \Trans 
\text{ such that } \loc' \in \Locs^n \} 
\end{align*}
and $\Trans^{n+1} = \{(\loc, \loc') ~|~ \loc \in \Locs^{n+1} \text{ and } \loc' \in 
V^n \}$. The attractor is the limit of theses sequences, i.e. 
$\Locs^*_C =\lim_{n \to +\infty} \Locs^n_C$, and $\Trans^*_C =\lim_{n \to +\infty} 
\Trans^n_C$. By definition, for all vertices $\loc \in \Locs^*_{\LocsT}$, 
$(\Locs^*_{\LocsT}, \loc, \Trans^*_{\LocsT})$ defines a 
reachability games where \ReachPl wins.

\checkprefgames*
\begin{proof}
We  provide the proof for player $\Pi=\SafePl$. The proof for player $\ReachPl$ goes analogously.

From $\game$ and the MD-strategy $\sigma$ for $\SafePl$, we obtain a transition system in which non-deterministic choices only take place in the location $\LocsReach$.
In this transition system, we can check in polynomial time whether there is a path reaching $C$ and $\LocsT$, i.e., whether there is a $\sigma$-play that reaches $C$ and $\LocsT$.
Furthermore, we can determine in polynomial time whether the player $\SafePl$ can win in the reachability game $\game_C$ that has the same arena as $\game$, but where the target set for $\ReachPl$ is $C$. To do so, we compute the attractor $A_C$ of locations from which $\ReachPl$ can enforce that $C$ is reached, which we need again later, in polynomial time. So, we can check the first two conditions of Definition \ref{def:game_conterfactual-cause} in polynomial time. If the check is successful, we continue as follows.

In order to check the counterfactual condition 3, we have to find the largest number $n$ such that there are strategies following $\sigma$ for $n$ steps, but still avoid $C$.
For this purpose, we construct a sequence of sets of locations that indicate where $\sigma$-plays can end after successively increasing number of steps $i$. We set $T_0=\{\locinit\}$.
Given $T_i$, we compute 
\begin{align*}
T_{i+1}=&\{t \in V \mid \text{ there is a $s\in T_i\cap \LocsSafe$ with $\sigma(s)=(s,t)$}\} \\
&\cup
\{ t\in V \mid \text{ there is a $s\in T_i\cap \LocsReach$ with $(s,t)\in \Trans$}\} \\
& \cup
(T_i\cap \LocsT) .
\end{align*}
Now, we check whether $T_{i+1}\cap A_C=\emptyset$. If this is the case, player $\SafePl$ can still avoid $C$ after following $\sigma$ for $i+1$ steps and we continue to compute $T_{i+2}$.
Otherwise, we have found that after following $\sigma$ for $i+1$ steps, it is not possible for $\SafePl$ to avoid $C$. Hence, we have found the largest number $n=i$ of steps after which this is still possible. Note that we find this $n$ after less than $V$-many steps as there is a simple $\sigma$-play reaching $C$ by condition 1, which has already been checked.
So, the sequence $T_0,\dots, T_n$ is computable in polynomial time.

Now, we have to check whether all MD-strategies $\tau$ for player $\SafePl$ that follow $\sigma$ for $n$ steps and avoid $C$ win in $\game$.
We construct a transition system: First, we let $T=\bigcup_{0\leq i\leq n}T_i$. At all these locations that are reachable within $n$ steps under $\sigma$, the strategies $\tau$ under consideration have to follow $\sigma$. So, we remove all transitions $(s,t)$ from a state $s\in T\cap \LocsSafe$ with $(s,t)\not = \sigma(s)$.
Furthermore, the strategies have to avoid $C$ and hence, we remove all locations from $A_C$ and their ingoing transitions. Let us call the set of remaining locations $S$ and the set of remaining transitions $\to$. Now, we check in polynomial time whether $\LocsT$ is reachable in the transition system $(S,\locinit,\to)$.
If this is not the case, $C$ is a $\dpref^H$ counterfactual cause for the fact that $\sigma$ loses in $\game$. Otherwise, $C$ is not such a cause.
To see this, note that any path in $(S,\locinit,\to)$ is a play under a strategy $\tau$ that follows $\sigma$ for $n$ steps and avoids $C$. Such a strategy belongs to the $\dpref^H$ closest strategies to $\sigma$ avoiding $C$. Conversely, any play under a strategy that follows $\sigma$ for $n$ steps and avoids $C$ is a path in $(S,\locinit,\to)$.
\end{proof}

\checkhammgames*

Before to make the proof, we recall some properties on shortest-path games.
A \emph{shortest-path games} is a game $\widetilde{\game} = (\widetilde{\Locs}, \locinit, \Trans, \weight)$
between players \MinPl and \MaxPl 
where $\game = (\widetilde{\Locs}, \locinit, \Trans)$ with 
$\widetilde{\Locs} = \LocsMin \uplus \LocsMax \uplus \Locs_T$ is a reachability game 
and $\weight : \Trans\to \Z$ is a weighted function on transitions of \game.
The objective of \MinPl is to reach $\Locs_T$, while 
minimising the weight of the play  \looseplay where
$\weight(\looseplay)=+\infty$ if \looseplay is infinite, and 
and $\weight(\looseplay) = \sum_{i=0}^{k-1} \weight(\loc_i,\loc_{i+1}$ if 
$\looseplay = \loc_0\cdots \loc_k$  is finite and $\loc_k \in \Locs_T$. 

The value
$\sCost(\loosestrategy) = \sup_{\winstrategy \mid \text{strategy of \MaxPl}} 
\weight(\outcomes(\locinit, \loosestrategy, \winstrategy))$ is the cost 
of a strategy \loosestrategy for $\MinPl$ where $\outcomes(\locinit, \loosestrategy, \winstrategy)$ is the unique 
\loosestrategy-play and \winstrategy-play  starting in \locinit. 
Then, the cost of the game is defined by 
$\cost = \inf_{\loosestrategy \mid \text{strategy of \MinPl}}
\sCost(\loosestrategy)$. In particular, a strategy \loosestrategy of \MinPl 
(resp. \MaxPl) is \emph{optimal} when $\sCost(\loosestrategy) \leq \cost$ (resp. 
$\sCost(\loosestrategy) \geq \cost$).
In shortest-path games \game with only non-negative weights, \cost 
is computable in polynomial time. Moreover, we can compute an optimal memoryless 
strategy for \MinPl and \MaxPl in polynomial time (see~\cite{KhachiyanBBEGRZ-07}).

\begin{proof}[Proof of Theorem~\ref{thm:checking-hamm-game}]
	We provide the proof for player $\Pl = \ReachPl$. The proof for player $\SafePl$ goes analogously. 
	
	As in the proof of Theorem~\ref{thm:checking-pref-game}, we can test in polynomial time the first two conditions 
	of Definition~\ref{def:game_conterfactual-cause}: the first one is given by testing reachability of $C$ in the 
	transition system defined by $\game$ and $\winstrategy$, the second one is given by computing an attractor to $C$ 
	for \SafePl. If the check is successful, we continue as follows. 
	
	In order to check the last counterfactual condition, we use a shortest-path game with only non-negative weights 
	that can be solved in polynomial time (see~\cite{KhachiyanBBEGRZ-07}). To do it, we define two shortest-path games\footnote{The idea of this proof is analogous that the proof of Theorem~\ref{thm:checking_dhamm}.} where 
	weights are only in $\{0, 1\}$.   
	In particular, we denote by $\widetilde{\game}_1 = (\widetilde{\Locs}_1, \locinit, \Trans, \weight)$  the shortest path 
	game such that $\widetilde{\Locs}_1  = \Locs \setminus A_C$ where $A_C$ is the set of locations in the attractor of $C$ for the player \SafePl, and for all transitions $(\loc, \loc') \in \Trans$, 
	$\weight(\loc, \loc') = 1$ if and only if $\loc \in \LocsReach$ and $\loosestrategy(\loc) \neq (\loc, \loc')$. Otherwise 
	the weight of a transition is $0$. Now, we fix $\weight_1$ be the value obtained in $\widetilde{\game}_1$ by \ReachPl. 
	Intuitively, $\cost_1$ is the minimal number of change in $\loosestrategy$ to obtain a winning strategy that avoids $C$. 
	Formally, let \strategy be such a strategy and \play be a \strategy-play, and by definition of $\widetilde{\game}_1$, 
	we have
	\begin{displaymath}
	\cost(\play) = |\{\loc_i ~|~ \loc_i \in \play \text{ and } 
	\strategy(\loc_i) \neq \loosestrategy(\loc_i)\}| = \dhamm^s(\strategy, \loosestrategy) \,.
	\end{displaymath}
	In particular, we obtain that  
	\begin{displaymath}
	\cost_1 = \min_{\substack{\strategy \\ \text{winning strategy}}} \sCost(\strategy) = 
	\sup_{\substack{\play \\ \text{\strategy-play}}} \dist(\play, \loosestrategy) \,.
	\end{displaymath} 
	
	To conclude, we need to check if all losing strategies that avoid $C$ are a distance to \loosestrategy greater that 
	$\cost_1$. To do it, we use a new shortest-path game. To define it, we note that, since $\game$ is acyclic, player \ReachPl 
	loses only when he reaches a final location, i.e. a location without any outgoing transition, that is not in $\LocsT$. 
	In the following, we denote this set by $\LocsNegT$. Now, we let $\widetilde{\game}_2 = (\widetilde{\Locs}_2, \locinit, \Trans, \weight)$ be the shortest-path game such that $\widetilde{\Locs}_2  = \widetilde{\Locs}_1$ such that $\widetilde{\Locs}_{\textsl{Eff}} = \LocsNegT$ and $\widetilde{\Locs}_{\neg\textsl{Eff}} = \LocsT$. 
	Moreover, for all transitions $(\loc, \loc') \in \Trans$, 
	$\weight(\loc, \loc') = 1$ if and only if $\loc \in \LocsReach$ and $\loosestrategy(\loc) \neq (\loc, \loc')$. Otherwise 
	the weight of a transition is $0$. Now, we fix $\weight_2$ be the value obtained in $\widetilde{\game}_2$ by \ReachPl. 
	Intuitively, $\cost_2$ is the minimal number of change in $\loosestrategy$ to obtain a losing strategy that avoids $C$. 
	Formally, by the link with the cost of a play and the distance $\dhamm^s$, we obtain that  
	\begin{displaymath}
	\cost_2 = \min_{\substack{\strategy \\ \text{losing strategy}}} \sCost(\strategy) = 
	\sup_{\substack{\play \\ \text{\strategy-play}}} \dist(\play, \loosestrategy) 
	\end{displaymath} 
	since we have reversed the winning condition in $\widetilde{\game}_2$.
	
	Finally, we check in polynomial time whether $\cost_1 < \cost_2$, otherwise $C$ is not a $\dhamm^s$-counterfactual cause. 
	Indeed, by interpreting the optimal cost in $\widetilde{\game}_1$ and $\widetilde{\game}_2$, we note that $\cost_1 < \cost_2$ 
	if and only if the strategies that avoid $C$ with a minimal distance to \loosestrategy is winning. 
\end{proof}

\section{Proof of Proposition~\ref{prop:dHH_NP-c}}
\label{app:distance-strategies}

In this appendix, we give the detailed proof of Proposition~\ref{prop:dHH_NP-c}.

\propdHH*

To prove this Proposition, we will prove that deciding if 
$\dstrat{\winstrategy}(\loosestrategy) \geq k$ is NP-complete. The NP-easiness 
is given by guessing a \winstrategy-play  \winplay such that 
$\dist(\winplay, \loosestrategy) \geq k$. This check can be done in polynomial time 
since we simply count the number of vertices along \winplay where \loosestrategy does 
not choose the same edge as \winstrategy. 

For the NP-hardness, we provide a reduction from the problem of the \emph{longest path problem} 
in directed graphs that is a NP-hard problem \cite{Schrijver-03}. In particular, 
given a directed graph $G$, the longest path problem ask if there exists a path 
in $G$ that visits at most one time each vertices of $G$ with a length at least 
$k$.

Let $G = (L, \Delta)$ be an instance of the longest path problem.
Intuitively, we define  a reachability game \game in which each vertex of $L$ is duplicated and one copy of the state belongs to each of the players. 
 In the vertex belonging to 
\Pl, there are the choices to go to the copy of the vertex belonging to the opponent or  to go 
 to the effect vertex. In the states belonging to the opponent $\neg\Pl$  edges are given by $\Delta$ and lead to the copies belonging to $\Pl$ (see 
Figure~\ref{fig:ex_longest-path}). Formally, we define \game such that 
$\Locs = \Locs_{\Pl} \cup \Locs_{\neg\Pl} \cup \{\locT\}$ where
$\Locs_{\Pl} = L\times\{0\}$, $\Locs_{\neg \Pl} = L\times\{1\} \cup \{\locinit\}$, and 
\begin{align*} 
\Trans 
&= \{((\loc,0), ),(\loc,1)) ~|~ \loc \in L \} \\
&\cup \{((\loc,0), \locT) ~|~ \loc \in L \} \\
&\cup \{((\loc,1), (\loc',0)) ~|~  (\loc, \loc') \in \Delta \} \\
&\cup \{(\locinit, (\loc,0)) ~|~ \loc \in L\}
\end{align*}
To conclude this reduction, we fix $\loosestrategy(\loc) = ((\loc,0),(\loc,1))$ 
and $\winstrategy(\loc) = ((\loc,0), \locT)$ for all vertices of \Pl. 
\figurename~\ref{fig:ex_longest-path} gives the reduction from the graph (i.e. an 
instance of the longest path problem) depicted on the left and the reachability 
game obtained with its non-winning MD-strategy (in green) and the winning MD-strategy 
(in pink) on the right. 
Now, we claim:

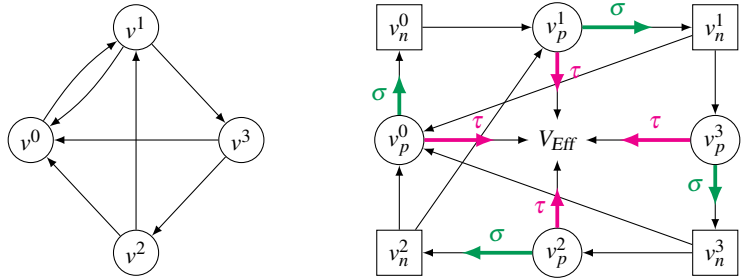
\begin{figure}[tbp]
	\centering
	\begin{tikzpicture}[xscale=.7,every node/.style={font=\footnotesize}, 
	every label/.style={font=\scriptsize}]
	\node[PlayerReach] at (0, 0) (v0) {$\loc^0$};
	\node[PlayerReach] at (2, 1.5) (v1) {$\loc^1$};
	\node[PlayerReach] at (2,-1.5) (v2) {$\loc^2$};
	\node[PlayerReach] at (4, 0) (v3) {$\loc^3$};
	
	% Connect the states with arrows
	\draw[->] 
	(v0) edge[bend left=10] (v1)
	(v1) edge[bend left=10] (v0)
	(v1) edge (v3)
	(v2) edge (v1)
	(v2) edge (v0)
	(v3) edge (v2)
	(v3) edge (v0);
	
	\begin{scope}[xshift=7cm]
	\node[PlayerReach] at (0, 0) (v0) {$\loc^0_p$};
	\node[PlayerReach] at (3, 1.5) (v1) {$\loc^1_p$};
	\node[PlayerReach] at (3,-1.5) (v2) {$\loc^2_p$};
	\node[PlayerReach] at (6, 0) (v3) {$\loc^3_p$};
	\node[PlayerSafe] at (0, 1.5) (s0) {$\loc^0_n$};
	\node[PlayerSafe] at (6, 1.5) (s1) {$\loc^1_n$};
	\node[PlayerSafe] at (0,-1.5) (s2) {$\loc^2_n$};
	\node[PlayerSafe] at (6, -1.5) (s3) {$\loc^3_n$};
	\node[target] at (3,0) (t) {$\LocsT$};
	
	\node[strat] at (0,1) (st0) {};
	\node[strat] at (5,1.5) (st1) {};
	\node[strat] at (1,-1.5) (st2) {};
	\node[strat] at (6,-1) (st3) {};
	
	\node[strat] at (2,0) (st01) {};
	\node[strat] at (3,.5) (st11) {};
	\node[strat] at (3,-.5) (st21) {};
	\node[strat] at (4,0) (st31) {};
	
	% Connect the states with arrows
	\draw[->] 
	(v0) edge (s0)
	(v0) edge (t)
	(v1) edge (s1)
	(v1) edge (t)
	(v2) edge (s2)
	(v2) edge (t)
	(v3) edge (s3)
	(v3) edge (t)
	(s0) edge (v1)
	(s1) edge (v0)
	(s1) edge (v3)
	(s2) edge (v1)
	(s2) edge (v0)
	(s3) edge (v2)
	(s3) edge (v0)
	;
	
	\draw[->,ForestGreen, line width=0.5mm] 
	(v0) edge node[left] {\textcolor{ForestGreen}{\loosestrategy}} (st0)
	(v1) edge node[above] {\textcolor{ForestGreen}{\loosestrategy}} (st1)
	(v2) edge node[above] {\textcolor{ForestGreen}{\loosestrategy}} (st2)
	(v3) edge node[left] {\textcolor{ForestGreen}{\loosestrategy}} (st3);
	
	\draw[->,Magenta, line width=0.5mm] 
	(v0) edge node[above, near end] {\textcolor{Magenta}{\winstrategy}} (st01)
	(v1) edge node[right] {\textcolor{Magenta}{\winstrategy}} (st11)
	(v2) edge node[left] {\textcolor{Magenta}{\winstrategy}} (st21)
	(v3) edge node[above] {\textcolor{Magenta}{\winstrategy}} (st31);
	\end{scope}
	\end{tikzpicture}
	\caption{On the left, we have a cyclic graph $G$ be an instance of the 
		longest path problem. On the right, we define \game be the reachability 
		game obtained by the reduction (without the initial vertex).}
	\label{fig:ex_longest-path}
\end{figure}

\begin{lemma}
	The following statements are equivalent:
	\begin{enumerate}
		\item in $G$, there exists a path \play that contains  each 
		vertex of $L$ at most once such that $|\play| \geq k$;
		\item in \game, $\dstrat{\winstrategy}(\loosestrategy) \geq k$.
	\end{enumerate}
\end{lemma}
\begin{proof}
	We suppose that there exists a path \play that contains at most one time each 
	vertex of $L$ such that $|\play| \geq k$. This path can be seen as a \loosestrategy-play  \looseplay that enters a loop in the last step.
	We observe  that $\dist(\looseplay, \winstrategy) = |\play| \geq k$. In particular, 
	by applying the supremum over \loosestrategy-plays, we conclude that 
	$\dstrat{\winstrategy}(\loosestrategy) \geq k$.
	
	Conversely, we suppose that $\dstrat{\winstrategy}(\loosestrategy) \geq k$. In 
	particular, there exists a \loosestrategy-play denoted \looseplay such that 
	$\dist(\looseplay, \loosestrategy) \geq k$. 
	Now, by projecting \looseplay in $G$, 
	i.e. by keeping only choices of vertices of the opponent, we obtain a path $\xi$ until a state is repeated for the first time (from then on the play loops) with 
	$|\play| = \dist(\looseplay, \loosestrategy) \geq k$.
\end{proof}

Now, we have tools to prove Proposition~\ref{prop:dHH_NP-c}.
\begin{proof}[Proof of Proposition~\ref{prop:dHH_NP-c}]
	To conclude, we need to prove that the reduction is  a polynomial-time reduction. 
	In particular, we remark that $|\Locs| \leq 2|L| + 1$, and  
	$|\Trans| \leq |\Delta| + 2|L|$. Thus, the reduction from $G$ into \game 
	with \loosestrategy and \winstrategy is given in polynomial time.
\end{proof}

\section{Proof of Theorem~\ref{thm:strat_finding-NP-c}}
\label{app_FVS-CSS}

In this appendix, we make the detailed proof of Theorem~\ref{thm:strat_finding-NP-c}.
\stratFindingNPc*

Now, we prove Theorem~\ref{thm:strat_finding-NP-c}. As explained in the main paper, 
we prove the NP upper bound for $ \dhamm^s$ since we can test the winning and compute 
the Hamming strategy distance in polynomial time.

 Hence, we will focus on 
the hardness proofs. To do that, we use polynomial-time many-one reduction to the problem for $\dhamm^s$ and a polynomial-time Turing reduction to the problem for \dHH from the decision 
version of the \emph{feedback vertex set} problem that is NP-complete by \cite{Karp1972}.
Given a cyclic (directed) graph $G$, the feedback vertex set problem  
consists of finding the minimal set of vertices of $G$, denoted $S$, such that if 
we remove this set, $G \setminus S$ becomes acyclic. Its decision version takes a 
threshold $k \in N$ and asks the existence of this set $S$ such that 
$|S| \leq k$.

\paragraph*{NP-hardness of the problem for $\dhamm^s$.}
We consider a (cyclic) directed graph $G = (L, \Delta)$ to be an instance of the 
feedback vertex set problem. Intuitively, from $G$, we define 
$\game = (\Locs, \locinit, \Trans)$ by adding a vertex of \ReachPl before each 
\SafePl that is given by $G$. Thus, vertices of \SafePl keep the possible choices 
from $G$, and vertices of \ReachPl can choose between staying in the game by the 
choice of \loosestrategy or reaching the effect (see Figure~\ref{fig:ex_feedback}).
Formally, we define \game such that 
$\Locs = \LocsReach \cup \LocsSafe \cup \{\locT\}$ where
$\LocsReach = L\times\{0\}$, $\LocsSafe = L\times\{1\} \cup \{\locinit\}$, and 
\begin{align*} 
\Trans 
&= \{((\loc,0), ),(\loc,1)) ~|~ \loc \in L \} \\
&\cup \{((\loc,0), \locT) ~|~ \loc \in L \} \\
&\cup \{((\loc,1), (\loc',0)) ~|~  (\loc, \loc') \in \Delta \} \\
&\cup \{(\locinit, (\loc,0)) ~|~ \loc \in L\}
\end{align*}
To conclude this reduction, we fix $\loosestrategy(\loc) = ((\loc,0), (\loc,1))$ 
for all vertices of \ReachPl. \figurename{~\ref{fig:ex_feedback}} gives the 
reduction from the graph (i.e. an instance of the feedback vertex set problem) 
depicted on the left and the reachability game obtained with its non-winning 
MD-strategy on the right. 
Now, we can compute this reduction is given in polynomial time since $|\Locs| \leq 2|L| + 1$, and  
$|\Trans| \leq |\Delta| + 2|L|$. Thus, the reduction from $G$ into \game 
with \loosestrategy is given in polynomial time. 
Finally, we prove the correctness of this reduction.

\begin{figure}[t]
	\centering
	\begin{tikzpicture}[xscale=.7,every node/.style={font=\footnotesize}, 
	every label/.style={font=\scriptsize}]
	\node[PlayerReach] at (0, 0) (v0) {$\loc^0$};
	\node[PlayerReach] at (2, 1) (v1) {$\loc^1$};
	\node[PlayerReach] at (2,-1) (v2) {$\loc^2$};
	\node[PlayerReach] at (4, 0) (v3) {$\loc^3$};
	
	% Connect the states with arrows
	\draw[->] 
	(v0) edge[bend left=10] (v1)
	(v1) edge[bend left=10] (v0)
	(v1) edge (v3)
	(v2) edge (v1)
	(v2) edge (v0)
	(v3) edge (v2)
	(v3) edge (v0);
	
	\begin{scope}[xshift=7cm]
	\node[PlayerReach] at (0, 0) (v0) {$\loc^0_r$};
	\node[PlayerReach] at (3, 1) (v1) {$\loc^1_r$};
	\node[PlayerReach] at (3,-1) (v2) {$\loc^2_r$};
	\node[PlayerReach] at (6, 0) (v3) {$\loc^3_r$};
	\node[PlayerSafe] at (0, 1) (s0) {$\loc^0_s$};
	\node[PlayerSafe] at (6, 1) (s1) {$\loc^1_s$};
	\node[PlayerSafe] at (0,-1) (s2) {$\loc^2_s$};
	\node[PlayerSafe] at (6, -1) (s3) {$\loc^3_s$};
	\node[target] at (3,0) (t) {$\LocsT$};
	
	\node[strat] at (0,.75) (st0) {};
	\node[strat] at (5,1) (st1) {};
	\node[strat] at (1,-1) (st2) {};
	\node[strat] at (6,-.75) (st3) {};
	
	% Connect the states with arrows
	\draw[->] 
	(v0) edge (s0)
	(v0) edge (t)
	(v1) edge (s1)
	(v1) edge (t)
	(v2) edge (s2)
	(v2) edge (t)
	(v3) edge (s3)
	(v3) edge (t)
	(s0) edge (v1)
	(s1) edge (v0)
	(s1) edge (v3)
	(s2) edge (v1)
	(s2) edge (v0)
	(s3) edge (v2)
	(s3) edge (v0)
	
	(v0) edge[ForestGreen, line width=0.5mm] node[left] 
	{\textcolor{ForestGreen}{\loosestrategy}} (st0)
	(v1) edge[ForestGreen, line width=0.5mm] node[above] 
	{\textcolor{ForestGreen}{\loosestrategy}} (st1)
	(v2) edge[ForestGreen, line width=0.5mm] node[above] 
	{\textcolor{ForestGreen}{\loosestrategy}} (st2)
	(v3) edge[ForestGreen, line width=0.5mm] node[left] 
	{\textcolor{ForestGreen}{\loosestrategy}} (st3);
	\end{scope}
	\end{tikzpicture}
	\caption{On the left, we have a cyclic graph $G$ as in the
		feedback vertex set problem. On the right, the game  \game  as in the reduction (without the initial vertex) is depicted.}
	\label{fig:ex_feedback}
\end{figure}
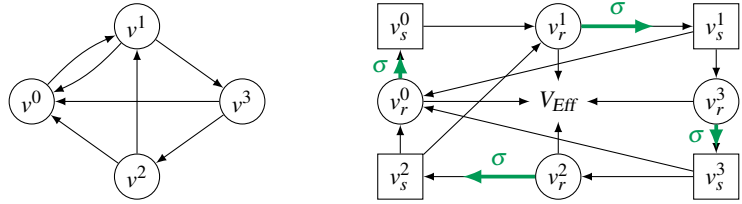

\begin{lemma}
	\label{lem:dhamm-NPhard}
	The both propositions are equivalent
	\begin{enumerate}
		\item in \game, there exists a winning MD-strategy for \ReachPl such 
		that $\dhamm^s(\winstrategy, \loosestrategy) \leq k$;
		\item in $G$, there exists a set of vertices $S$ such that $|S| \leq k$ 
		and $G \setminus S$ is acyclic.
	\end{enumerate}
\end{lemma}
\begin{proof}
	We suppose that there exists a winning MD-strategy \winstrategy for \ReachPl 
	player such that $\dhamm^s(\winstrategy, \loosestrategy) \leq k$ in 
	\game. From \winstrategy we define $S \subseteq \LocsReach$ be the set 
	of vertices where \winstrategy and \loosestrategy are distinct, i.e. 
	$S = \{\loc ~|~ \winstrategy(\loc) \neq \loosestrategy(\loc) \text{ and } 
	\loc \in \LocsReach\}$. By definition of the Hamming strategy distance, we know 
	that $|S| \leq k$. To conclude, we prove that $G \setminus S$ is 
	acyclic. By contradiction, we suppose that $c = \loc_0 \loc_1 \cdots \loc_n 
	\loc_0$ be a cycle of $G \setminus S$. By definition of $S$, we know that 
	for all vertices of $c$ that belong to \ReachPl satisfy 
	$\winstrategy(\loc) = \loosestrategy(\loc)$. In particular, there exists a 
	\winstrategy-play that is cyclic and that never reaches the effect. Thus, 
	\winstrategy is non winning.
	
	Conversely, we suppose that there exists a set of vertex $S$ such that 
	$|S| \leq k$ and $G \setminus S$ is acyclic. From $S$, we define an 
	MD-strategy for \ReachPl such that, for all $\loc \in \LocsReach$, 
	we have
	\begin{displaymath}
	\winstrategy(\loc_s) = 
	\begin{cases}
	\loosestrategy(\loc_s) & \text{if $\loc \notin S$;} \\
	\locT & \text{otherwise.}
	\end{cases}
	\end{displaymath}
	Since the graph $G \setminus S$ is acyclic, then all \winstrategy-play 
	reach the effect (by definition of \game). Thus, this MD-strategy is winning. 
	Moreover, the MD-strategy satisfy that $\dhamm^s(\winstrategy, \loosestrategy) 
	\leq k$ since the vertices where \winstrategy and \loosestrategy are  
	distinct is exactly the vertices of $S$. By hypothesis, $|S| \leq k$, thus 
	$\dhamm^s(\winstrategy, \loosestrategy) = |S| \leq k$.
\end{proof}

\paragraph*{Polynomial-time Turing reduction to the problem for \dHH.}
To conclude the proof of Theorem~\ref{thm:strat_finding-NP-c}, we need to 
adapt the previous reduction to the case of \dHH. However, the reduction can not 
directly apply in the case of \dHH. Indeed, when we define $S$ from a winning 
MD-strategy with a distance less than $k$, we can not guarantee its size (see 
the following example).

\begin{figure}[tbp]
	\centering
	\begin{tikzpicture}[xscale=.7,every node/.style={font=\footnotesize}, 
	every label/.style={font=\scriptsize}]
	\node[PlayerReach] at (0, 1.5) (v0) {$\loc^0$};
	\node[PlayerReach] at (0, 0) (v1) {$\loc^1$};
	\node[PlayerReach] at (2,-1.5) (v2) {$\loc^2$};
	\node[PlayerReach] at (2, 0) (v3) {$\loc^3$};
	\node[PlayerReach] at (-2,-1.5) (v4) {$\loc^4$};
	\node[PlayerReach] at (-2, 0) (v5) {$\loc^5$};
	
	% Connect the states with arrows
	\draw[->] 
	(v0) edge[bend left=10] (v1)
	(v1) edge[bend left=10] (v0)
	(v2) edge[bend left=10] (v3)
	(v3) edge[bend left=10] (v2)
	(v4) edge[bend left=10] (v5)
	(v5) edge[bend left=10] (v4)
	(v1) edge (v2)
	(v1) edge (v4);
	
	\begin{scope}[xshift=10cm]
	\node[PlayerReach] at (-1, 1.5) (v0) {$\loc^0_r$};
	\node[PlayerReach] at (0, 0) (v1) {$\loc^1_r$};
	\node[PlayerReach] at (2,-1.5) (v2) {$\loc^2_r$};
	\node[PlayerReach] at (2, 0) (v3) {$\loc^3_r$};
	\node[PlayerReach] at (-2,-1.5) (v4) {$\loc^4_r$};
	\node[PlayerReach] at (-2, 0) (v5) {$\loc^5_r$};
	\node[PlayerSafe] at (1, 1.5) (s0) {$\loc^0_s$};
	\node[PlayerSafe] at (0, -1.5) (s1) {$\loc^1_s$};
	\node[PlayerSafe] at (4,-1.5) (s2) {$\loc^2_s$};
	\node[PlayerSafe] at (4, 0) (s3) {$\loc^3_s$};
	\node[PlayerSafe] at (-4,-1.5) (s4) {$\loc^4_s$};
	\node[PlayerSafe] at (-4, 0) (s5) {$\loc^5_s$};
	%\node[target] at (3,0) (t) {$\LocsT$};
	
	\node[strat] at (0.5,1.5) (st0) {};
	\node[strat] at (0,-1) (st1) {};
	\node[strat] at (3.5,-1.5) (st2) {};
	\node[strat] at (3.5,0) (st3) {};
	\node[strat] at (-3.5,-1.5) (st4) {};
	\node[strat] at (-3.5,0) (st5) {};
	
	% Connect the states with arrows
	\draw[->] 
	(v0) edge (s0)
	(v1) edge (s1)
	(v2) edge (s2)
	(v3) edge (s3)
	(v4) edge (s4)
	(v5) edge (s5)
	(s0) edge (v1)
	(s1) edge[bend left=10] (v0)
	(s1) edge (v2)
	(s1) edge (v4)
	(s2) edge (v3)
	(s3) edge (v2)
	(s4) edge (v5)
	(s5) edge (v4)
	
	(v0) edge[ForestGreen, line width=0.5mm] node[above] 
	{\textcolor{ForestGreen}{\loosestrategy}} (st0)
	(v1) edge[ForestGreen, line width=0.5mm] node[right] 
	{\textcolor{ForestGreen}{\loosestrategy}} (st1)
	(v2) edge[ForestGreen, line width=0.5mm] node[above] 
	{\textcolor{ForestGreen}{\loosestrategy}} (st2)
	(v3) edge[ForestGreen, line width=0.5mm] node[above] 
	{\textcolor{ForestGreen}{\loosestrategy}} (st3)
	(v4) edge[ForestGreen, line width=0.5mm] node[above] 
	{\textcolor{ForestGreen}{\loosestrategy}} (st4)
	(v5) edge[ForestGreen, line width=0.5mm] node[above] 
	{\textcolor{ForestGreen}{\loosestrategy}} (st5);
	\end{scope}
	\end{tikzpicture}
	\caption{On the left, we have a cyclic graph $G$ be an instance of the 
		feedback vertex set. On the right, we define \game be the reachability 
		game obtained with the reduction for Hamming strategy distance (without 
		the initial states).}
	\label{fig:ex_feedback-wrong}
\end{figure}
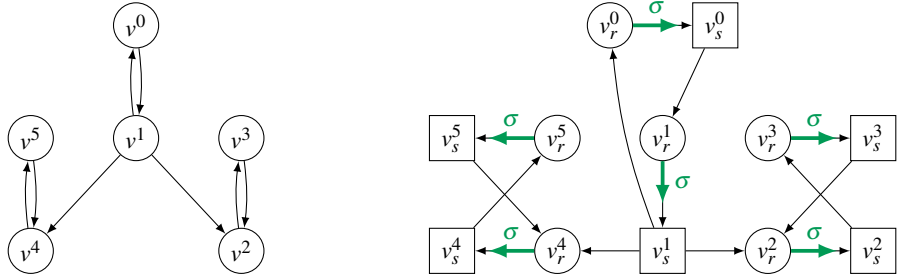

\begin{example}[Counter-example to the reduction for \dHH]
	\label{ex:feedbcak-wrong}
	We consider the directed graph $G$ depicted in the left of 
	Figure~\ref{fig:ex_feedback-wrong}. The size of a minimal feedback vertex set 
	in $G$ is $3$ (like we need to break the three independent elementary 
	cycles). 
	
	Now, we consider the reachability game \game and the non-winning strategy 
	\loosestrategy for \ReachPl depicted in the right of 
	\figurename{~\ref{fig:ex_feedback-wrong}}. In \game, a winning MD-strategy 
	\winstrategy that minimises \dHH satisfies 
	$\dHH(\loosestrategy, \winstrategy) = 2$ since 
	$\dstrat{\winstrategy}(\loosestrategy) = 1$ and 
	$\dstrat{\loosestrategy}(\winstrategy) = 2$ to cut the two cycles visiting by 
	some \winstrategy-play. Thus, \dHH between the both MD-strategies is less than 
	the minimal feedback vertex set.
	\markend
\end{example}

\begin{remark}
	If we extend the bound on the size of $S$ to avoid this behaviour, we have a 
	problem when $G$ contains only an SCC (like there exists a 
	\loosestrategy-play that reaches all vertices of \game). Indeed, in this 
	case, we can not define \winstrategy with the good distance from a set $S$.
\end{remark}

However, we can adapt the reduction in the case of \dHH by considering only 
instances of the feedback vertex set problem  with only one SCC. 
We consider a directed graph $G = (L, \Delta)$ to be an SCC that is an instance 
of the feedback vertex set problem. From the previous reduction, we define 
$\game = (\Locs, \locinit, \Trans)$ with a non-winning MD-strategy. This reduction 
is depicted in Figure {\ref{fig:ex_feedback}}. By the same arguments as for $\dhamm^s$, 
we prove that the reduction is given in polynomial time. Moreover, we can adapt the proof of 
Lemma~\ref{lem:dhamm-NPhard} to prove the correctness of the reduction for \dHH.

\begin{lemma}
	\label{lem:dHH-NPhard}
	The two following statements are equivalent
	\begin{enumerate}
		\item in \game, there exists a winning MD-strategy for \ReachPl such that 
		$\dHH(\winstrategy, \loosestrategy) \leq k$;
		\item in $G$, there exists a set of vertex $S$ such that $|S| \leq k$ and 
		$G \setminus S$ is acyclic.
	\end{enumerate}
\end{lemma}
\begin{proof}
	We suppose that there exists a winning MD-strategy \winstrategy for 
	\ReachPl such that $\dHH(\winstrategy, \loosestrategy) \leq k$ in 
	\game. From \winstrategy we define $S \subseteq \LocsReach$ be the set 
	of vertices where \winstrategy and \loosestrategy are distinct, i.e. 
	$S = \{\loc ~|~ \winstrategy(\loc) \neq \loosestrategy(\loc) \text{ and } 
	\loc \in \LocsReach\}$. Since $G$ is a SCC, there exists a 
	\loosestrategy-play \looseplay such that \looseplay reaches all 
	vertices of \game. By hypothesis on the distance between \winstrategy 
	and \loosestrategy, we know that $\dist(\looseplay, \winstrategy) \leq k$. 
	Thus, we deduce that the number of vertices of \game such that 
	\winstrategy and \loosestrategy are distinct is at least $k$, and 
	$|S| \leq k$. To conclude, we prove that $G \setminus S$ is acyclic. By 
	contradiction, we suppose that $c = \loc_0 \loc_1 \cdots \loc_n \loc_0$ be a 
	cycle of $G \setminus S$. By definition of $S$, we know that for all 
	vertices of $c$ that belong to \ReachPl satisfy 
	$\winstrategy(\loc) = \loosestrategy(\loc)$. In particular, there exists a 
	\winstrategy-play that is cyclic and that never reaches the effect. Thus, 
	\winstrategy is non winning.
	
	Conversely, we suppose that there exists a set of vertex $S$ such that 
	$|S| \leq k$ and $G \setminus S$ is acyclic. From $S$, we define an 
	MD-strategy for \ReachPl such that for all $\loc_r \in \LocsReach$, 
	we have 
	\begin{displaymath}
	\winstrategy(\loc_r) = 
	\begin{cases}
	\loosestrategy(\loc_r) & \text{if $\loc \notin S$;} \\
	\locT & \text{otherwise.}
	\end{cases}
	\end{displaymath}
	Since the graph $G \setminus S$ is acyclic, then all \winstrategy-play 
	reach the effect. Otherwise, the graph $G \setminus S$ will contains a cycle. 
	Thus, this MD-strategy is winning. Moreover, the MD-strategy satisfy that 
	$\dHH(\winstrategy, \loosestrategy) \leq k$ since 
	$\dstrat{\winstrategy}(\loosestrategy) = 1$ (if \winstrategy does not 
	choose the same transition than \loosestrategy, then the play reaches 
	$\LocsT$), and $\dstrat{\winstrategy}(\loosestrategy) \leq k$ since the worst 
	case is when there exists a \looseplay that be a \loosestrategy-play 
	and that reaches all vertices of $G$, in this case $\dist(\looseplay, 
	\winstrategy) \leq k$ by definition of $S$. 
\end{proof}

To conclude the proof of the theorem, we 
will show that the feedback vertex set problem is not in P if P$\not=$NP when we 
restrict the family of graphs to graphs that contain only an SCC.

\paragraph*{Feedback vertex set problem on SCCs.}
Formally, the feedback vertex set problem on SCCs is defined by a directed 
graph $G$ that is strongly connected and a threshold $k \in \N$. In this problem, we 
ask if there exists a set of vertices of $G$, denoted $S$, such that $|S|\leq k$ 
and $G \setminus S$ is acyclic. We provide a polynomial-time Turing reduction from the general feedback vertex set problem to the feedback vertex set problem on SCCs.

\begin{proposition}
	\label{prop:FVS-SCC-NP-h}
	The feedback vertex set problem on SCC is not in P if P$\not=$NP.
\end{proposition}

To prove this result, we consider the optimisation version 
of feedback vertex set problems on SCC or not, i.e. finding the size of a set $S$ with the 
minimal size. In particular, we prove that to solve the feedback vertex set 
optimization problem, we can solve the problem on each SCC independently. Then, we will use the fact
that we can solve  the optimisation problem on SCCs using the decision version via a binary search algorithm.

\begin{lemma}
	\label{lem:FVS-SCC-NP-h}
	The following two problems are polynomial-time Turing-equivalent: 
	\begin{enumerate}
		\item computing the size of a minimal feedback vertex set in an (arbitrary) undirected graph;
		\item computing  the size of the minimal feedback vertex set in each SCC 
		of a directed graph.
	\end{enumerate}
\end{lemma}
\begin{proof}
	We suppose that we can solve the minimal feedback vertex set. Let $G$ be 
	a directed graph, and with Tarjan's algorithm, for example, we compute all SCC 
	of $G$ in polynomial time. Now, since each SCC is a directed graph, then in 
	each SCC of $G$, we can solve the minimal vertex set.
	
	Conversely, we suppose that we can solve independently the feedback vertex set 
	in each SCC of all directed graph. Let $G$ be a directed graph. We compute 
	all SCCs of $G$ in polynomial time with Tarjan's algorithm for example. Let 
	$S_C$ be the solution of the minimal feedback vertex set problem for the 
	SCC $C$ of $G$. We fix $S = \cup_{C}~ S_C$ be the set of vertices that contains 
	all solutions for each SCC $C$. We note that $S$ is a solution for the 
	feedback vertex set problem of $G$ like each SCC is acyclic and the graph of 
	SCC is acyclic too. Now, $S$ is also minimal size, since all $S_C$ are a 
	minimal size.
\end{proof}

Now, we can do the proof of Proposition~\ref{prop:FVS-SCC-NP-h}.

\begin{proof}[Proof of Proposition~\ref{prop:FVS-SCC-NP-h}]
	We now describe a Turing reduction from the optimization version of the feedback 
	vertex set problem on a strongly connected graph $G$ to the decision version of this problem:
	We compute  the size $K$ of the minimal feedback vertex set by a binary search 
	between $0$ and $|G|$. This requires only polynomially many calls to a subroutine for the decision version of the feedback vertex set problem.
	 Thus, we can find the size of a minimal feedback vertex set on an SCC in polynomial time using an oracle for the decision version on SCCs. Now, by Lemma~\ref{lem:FVS-SCC-NP-h}, we know that we can use this computation on each SCC to compute the size of a minimal feedback vertex set on arbitrary graphs.
	 So, if the decision version of the feedback vertex set were in P, we could compute the size of a minimal feedback vertex set on arbitrary graphs and hence also solve the NP-hard decision version. So, if P$\not=$NP, the decision version of the feedback vertex set problem on SCCs cannot be solved in polynomial time.
\end{proof}

%\section{Proof of Corollary~\ref{cor:check-hamm-game}}
%\label{app:proof_cor-check-hamm-game}

\checkdhammsstrat*
\begin{proof}
	Checking if $E$ is an explanation can be done in polynomial time by fixing the decisions of \loosestrategy in states not in $E$ and removing the choices of \loosestrategy in state in $E$ and solving the resulting game. If it is an explanation, we know that $\dhamm^s(\tau,\sigma)=|E|$ for any $E$-distinct $\sigma$-strategy $\tau$.
	To check that $E$ is not a $\dhamm^s$-minimal explanation, we can then guess a winning strategy for $\ReachPl$ in $\game$ with $\dhamm^s(\tau,\sigma)<|E|$. So, the problem is in coNP.
	
	For the coNP-hardness,
	we provide a reduction from the problem for $\dhamm^s$ shown to be NP-complete in Theorem \ref{thm:strat_finding-NP-c} to the problem whether $E$ is not a $\dhamm^s$-minimal explanation.
	Let $\game$ with initial vertex $\locinit$, \loosestrategy, and $k$ be given as in Theorem~\ref{thm:strat_finding-NP-c}. We add $k+2$ new vertices $w_0,\dots, w_{k+1}$. 
	From $w_i$ with $i\leq k$, there are transitions to $w_{i+1}$ and to $\locinit$ and $w_i$ is controlled by $\ReachPl$. The vertex $w_{k+1}$ is a new additional target state, $w_0$ is the new initial vertex.
	Call the new game $\game^\prime$.
	We let $E=\{w_0,\dots,w_k\}$ and $\sigma^\prime$ be $\sigma$ extended with $\sigma^\prime(w_i)=(w_i,\locinit)$ for all $i$. An illustration of this construction is given 
	in Figure  \ref{fig:game_strat-cause-2}.
	
	Now, we claim that $E$ is not a $\dhamm^s$-minimal explanation if and only if there is a winning strategy $\tau$ for $\ReachPl$ in $\game$ with $\dhamm^s(\tau,\loosestrategy)\leq k$.
	Note that there is only one $E$-distinct $\sigma^\prime$-strategy $\tau^\prime$ in $\game^\prime$: this strategy chooses $(w_i,w_{i+1})$ in $w_i$ for $i\leq k$ and is winning for $\ReachPl$. Clearly, 
	$\dhamm^s(\tau^\prime,\loosestrategy^\prime)=k+1$. If there is a winning strategy $\tau$ for $\ReachPl$ in $\game$ with $\dhamm^s(\tau,\loosestrategy)\leq k$, then this strategy together with the choices $(w_i,\locinit)$ for $i\leq k$ is winning in $\game^\prime$ and has distance at most $k$ to $\sigma^\prime$.
	Conversely, any winning strategy $\tau^\prime$ in $\game^\prime$ with $\dhamm^s(\tau^\prime,\loosestrategy^\prime)\leq k$ cannot induce the play $w_0,\dots, w_{k+1}$. Hence, it has to move to $\locinit$ and consequently contain a winning strategy in $\game$ with a $\dhamm^s$-distance to $\sigma$ of at most $k$.
	
	The  proof that checking the \dHH-minimality of an explanation is not simpler than 
	the threshold problem for the distance to a winning strategy in Theorem~\ref{thm:strat_finding-NP-c} works analogously.
\end{proof}

\section{Proof of Theorem~\ref{thm:find-trans_poly}}
\label{app:poly-algo}

In this appendix, we give the detailed proof of Theorem~\ref{thm:find-trans_poly}.
\findTransPoly*

\paragraph*{Step 1: the shortest-path game reduction.}
Our goal is to compute a MD-strategy that is winning and minimises the 
distance $\dstrat{\strategy}(\loosestrategy)$. To compute this strategy, we 
use a shortest-path game to compare the distance of the different winning 
MD-strategies (that exist by hypothesis) according to the non-winning MD-strategy. 
Thus, shortest-path games allow us to find the ”best” one according to 
$\dstrat{\strategy}(\loosestrategy)$.

Formally, from a reachability game $\game = (\Locs, \locinit, \Trans)$, we define 
a shortest-path game $\widetilde{\game} = (\Locs, \locinit, \Trans, \weight)$ 
such that, for all transitions $(\loc, \loc') \in \Trans$, we fix 
\begin{displaymath}
\weight(\loc, \loc') = 
\begin{cases}
0 & \text{if $\loc \in \LocsSafe$, or if 
	$\loosestrategy(\loc) = (\loc, \loc')$;} \\
1 & \text{otherwise, i.e. if $\loc \in \LocsReach$ and 
	$\loosestrategy(\loc) \neq (\loc, \loc')$.}
\end{cases}
\end{displaymath}
This reduction (with the non-winning MD-strategy \loosestrategy is depicted in green) 
is illustrated in \figurename{~\ref{fig:ex_G-computation}}. Moreover, an MD-strategy that 
minimises $\dstrat{\strategy}(\loosestrategy)$ with \loosestrategy is \winstrategy 
depicted in pink.

\begin{figure}[tbp]
	\centering
	\begin{tikzpicture}[xscale=.7,every node/.style={font=\footnotesize}, 
	every label/.style={font=\scriptsize}]
	\node[PlayerReach] at (3, 1.5) (v1) {$\loc_1$};
	\node[PlayerSafe] at (0, 0) (v2) {$\loc_2$};
	\node[PlayerReach] at (3, -1.5) (v3) {$\loc_3$};
	\node[PlayerSafe] at (5, 0) (v0) {$\loc_0$};
	\node[target] at (0, 1.5) (t) {$\locT$};
	\node[target] at (0, -1.5) (t1) {$\locT$};
	
	\node[strat] at (1.2,.6) (st11) {};
	\node[strat] at (1,1.5) (st12) {};
	\node[strat] at (3,0) (st31) {};
	\node[strat] at (1,-1.5) (st32) {};
	
	% Connect the states with arrows
	\draw[->] 
	(v0) edge (v1)
	(v0) edge (v3)
	(v1) edge (t)
	(v1) edge (v2)
	(v2) edge (v3)
	(v3) edge (v1)
	(v3) edge (t1)
	;
	
	\draw[->,ForestGreen, line width=0.5mm]
	(v1) edge node[below] {\textcolor{ForestGreen}{\loosestrategy}} (st11)
	(v3) edge node[left] {\textcolor{ForestGreen}{\loosestrategy}} (st31)
	;
	
	\draw[->,Magenta, line width=0.5mm]
	(v1) edge node[above,yshift=.1cm] {\textcolor{Magenta}{\winstrategy}} (st12)
	(v3) edge node[below,yshift=-.1cm] {\textcolor{Magenta}{\winstrategy}} (st32);
	
	\begin{scope}[xshift=9cm]
	\node[PlayerReach] at (3, 1.5) (v1) {$\loc_1$};
	\node[PlayerSafe] at (0, 0) (v2) {$\loc_2$};
	\node[PlayerReach] at (3, -1.5) (v3) {$\loc_3$};
	\node[PlayerSafe] at (5, 0) (v0) {$\loc_0$};
	\node[target] at (0, 1.5) (t) {$\locT$};
	\node[target] at (0, -1.5) (t1) {$\locT$};
	
	% Connect the states with arrows
	\draw[->] 
	(v0) edge node[above] {$0$} (v1)
	(v0) edge node[below] {$0$} (v3)
	(v1) edge node[above] {$1$} (t)
	(v1) edge node[above] {$0$} (v2)
	(v2) edge node[above] {$0$} (v3)
	(v3) edge node[left] {$0$} (v1)
	(v3) edge node[below] {$1$} (t1);
	\end{scope}
	\end{tikzpicture}
	\caption{The transformation into $\widetilde{\game}$ (on the right) the shortest-path game 
		define to compute a winning strategy that minimises
		$\dstrat{\strategy}(\loosestrategy)$ from the reachability game \game 
		(on the left) and the non winning strategy \loosestrategy.}
	\label{fig:ex_G-computation}
\end{figure}
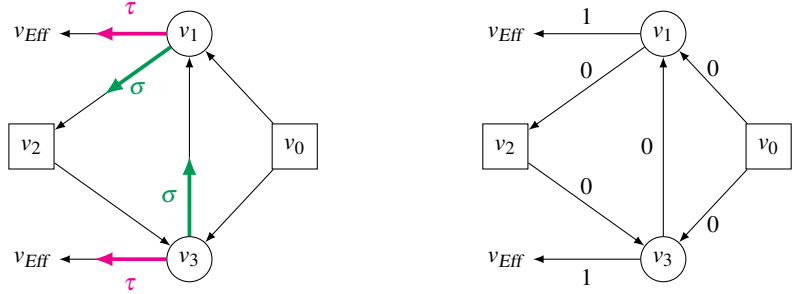

Now, we prove the correctness of this reduction by proving that an optimal 
MD-strategy in $\widetilde{\game}$ defines a winning strategy in \game such that 
it minimises $\dstrat{\strategy}(\loosestrategy)$. Let \winstrategy be an optimal 
MD-strategy in $\widetilde{\game}$. When we see \winstrategy like a MD-strategy 
in \game, we obtain the following Lemma.

\begin{lemma}
	\label{lem:DHH-strat_find} 
	Let \winstrategy be an optimal MD-strategy in the shortest-path game 
	$\widetilde{\game}$. Then, in the reachability game \game, \winstrategy is a 
	winning MD-strategy that minimises $\dstrat{\strategy}(\loosestrategy)$, i.e.
	$\dstrat{\winstrategy}(\loosestrategy) = 
	\min \{\dstrat{\strategy}(\loosestrategy) \mid 
	\strategy \text{ be a winning MD-strategy}\}$.
\end{lemma}
\begin{proof}
	Since there exists a winning MD-strategy in \game, then $\sCost$ of \locinit 
	is finite, i.e. $\sCost(\winstrategy) < +\infty$. Thus, \winstrategy is a 
	winning MD-strategy in \game.
	
	To conclude the proof, we prove that \winstrategy minimises 
	$\dstrat{\strategy}(\loosestrategy)$. To do that, we reason by contradiction, 
	and we suppose that there exists \strategy be a winning MD-strategy for \ReachPl 
	in \game such that $\dstrat{\strategy}(\loosestrategy) < 
	\dstrat{\winstrategy}(\loosestrategy)$.
	
	Let \play be a \strategy-play, and by definition of $\widetilde{\game}$, 
	$\cost(\play)$ denotes the number the occurrence of vertices of \ReachPl 
	along \play such that \strategy and \loosestrategy make distinct choices. 
	Formally, we have
	\begin{displaymath}
	\cost(\play) = |\{\loc_i ~|~ \loc_i \in \play \text{ and } 
	\strategy(\loc_i) \neq \loosestrategy(\loc_i)\}|.
	\end{displaymath}
	However, like \strategy is a winning MD-strategy, \play is acyclic. 
	In particular, $\cost(\play)$ denotes the number of vertices of \ReachPl 
	along \play such that \strategy and \loosestrategy make distinct choices. 
	Thus, by definition of the distance $\dist$ between a play and a MD-strategy, 
	we have $\cost(\play) = \dist(\play, \loosestrategy)$, and by 
	applying the same reasoning for \winstrategy, we obtain that  
	\begin{displaymath}
	\sCost(\strategy) = \sup_{\substack{\play \\ \text{\strategy-play}}} 
	\dist(\play, \loosestrategy) 
	\qquad \text{and} \qquad \sCost(\winstrategy) = 
	\sup_{\substack{\winplay \\ \text{\winstrategy-play}}} 
	\dist(\winplay, \loosestrategy).
	\end{displaymath} 
	Now, since \winstrategy is an optimal MD-strategy in $\widetilde{\game}$, we 
	deduce that 
	\begin{displaymath}
	\dstrat{\strategy}(\loosestrategy) = \sCost(\strategy) \geq 
	\sCost(\winstrategy) = \dstrat{\winstrategy}(\loosestrategy).
	\end{displaymath} 
	Finally, we apply the hypothesis over distance of \strategy, and we obtain a 
	contradiction since $\dstrat{\strategy}(\loosestrategy) \geq  
	\dstrat{\winstrategy}(\loosestrategy) > 
	\dstrat{\strategy}(\loosestrategy)$.
\end{proof}

\begin{remark}
	This transformation is correct for all reachability games and non-winning 
	MD-strategies. In particular, we can always compute a winning MD-strategy 
	to minimises $\dstrat{\strategy}(\loosestrategy)$ in polynomial time. 
\end{remark}

The previous Lemma implies that an optimal MD-strategy in $\widetilde{\game}$ 
is winning and minimises $\dstrat{\strategy}(\loosestrategy)$ in \game. 
Nevertheless, this strategy does not minimise the distance $\dHH$ since it 
does not minimise $\dstrat{\loosestrategy}(\strategy)$ (see 
following example). In particular, we want to optimise 
\winstrategy according to $\dstrat{\loosestrategy}(\winstrategy)$ without 
changing $\dstrat{\winstrategy}(\loosestrategy)$. This transformation is given 
by the following steps of the algorithm.

\begin{example}
	\label{ex:dHH_looseplay}
	We consider \game be the reachability game depicted on the left of 
	\figurename{~\ref{fig:ex_G-computation}} and \loosestrategy be the non-winning 
	MD-strategy for \ReachPl depicted in green. Moreover, we consider \winstrategy, 
	depicted in pink, be a winning MD-strategy for \ReachPl that minimises 
	$\dstrat{\winstrategy}(\loosestrategy)$ 
	(but not for $\dHH$). Now, we compute \dHH by computing both terms. For the 
	first one, we have  
	$\dstrat{\winstrategy}(\loosestrategy)  
	= \max(\dist(\loc_0 \, \loc_1 \, \locT, \loosestrategy), 
	\dist(\loc_0 \, \loc_3 \, \locT, \loosestrategy)) = 1 $
	and $\dstrat{\loosestrategy}(\winstrategy) = 
	\dist((\loc_1 \, \loc_2 \, \loc_3)^{\omega}, \loosestrategy) = 2$ for the second one. 
	Thus, $\dHH(\loosestrategy, \winstrategy) = 2 = \dstrat{\loosestrategy}(\winstrategy) 
	> \dstrat{\winstrategy}(\loosestrategy)$.
	\markend
\end{example}

\paragraph*{Step 2: the set of connected component reduction.}
The goal of this second step is to define another winning MD-strategy in \game 
from \winstrategy computed in the previous step such that this new MD-strategy 
minimises the distance $\dstrat{\strategy}(\loosestrategy)$ without increasing 
the distance $\dstrat{\loosestrategy}(\winstrategy^1) = 
\dstrat{\loosestrategy}(\winstrategy)$. To do that, we can synchronise choices 
of \winstrategy with ones of \loosestrategy for all vertices where any 
\winstrategy-plays can go, i.e. for vertices that are not reachable under 
\winstrategy against all strategies of \SafePl. With this synchronisation, we 
want to decrease the distance between all \loosestrategy-play that follows any 
choice of \winstrategy (see Example~\ref{ex:synch-1}) so that, under the 
hypothesis on \game and \loosestrategy, we prove that this new MD-strategy 
minimises $\dHH$ according to \loosestrategy.

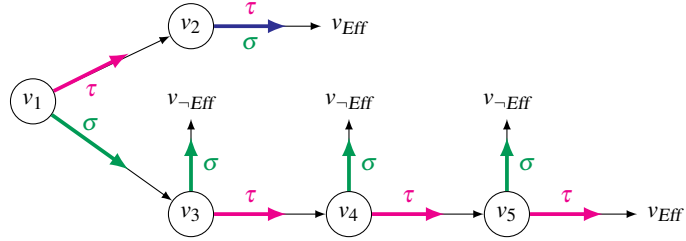
\begin{figure}[tbp]
	\centering
	\begin{tikzpicture}[xscale=.7,every node/.style={font=\footnotesize}, 
	every label/.style={font=\scriptsize}]
	\node[PlayerReach] at (0,0) (v1) {$\loc_1$};
	\node[PlayerReach] at (3, 1) (v2) {$\loc_2$};
	\node[PlayerReach] at (3,-1.5) (v3) {$\loc_3$};
	\node[PlayerReach] at (6, -1.5) (v4) {$\loc_4$};
	\node[PlayerReach] at (9, -1.5) (v5) {$\loc_5$};
	\node[target] at (6,1) (t1) {$\locT$};
	\node[target] at (3,0) (t2) {$\locNegT$};
	\node[target] at (6,0) (t3) {$\locNegT$};
	\node[target] at (9,0) (t4) {$\locNegT$};
	\node[target] at (12,-1.5) (t5) {$\locT$};
	
	\node[strat] at (2,-1) (st11) {};
	\node[strat] at (2,.7) (st12) {};
	\node[strat] at (5, 1) (st2) {};
	\node[strat] at (3,-.35) (st31) {};
	\node[strat] at (5,-1.5) (st32) {};
	\node[strat] at (6,-.35) (st41) {};
	\node[strat] at (8,-1.5) (st42) {};
	\node[strat] at (9,-.35) (st51) {};
	\node[strat] at (11,-1.5) (st52) {};
	
	% Connect the states with arrows
	\draw[->] 
	(v1) edge (v2)
	(v1) edge (v3)
	(v2) edge (t1)
	(v3) edge (v4)
	(v3) edge (t2)
	(v4) edge (v5)
	(v4) edge (t3)
	(v5) edge (t4)
	(v5) edge (t5)
	;
	
	\draw[->,ForestGreen, line width=0.5mm]
	(v1) edge node[above] {\textcolor{ForestGreen}{\loosestrategy}} (st11)
	(v3) edge node[right] {\textcolor{ForestGreen}{\loosestrategy}} (st31)
	(v4) edge node[right] {\textcolor{ForestGreen}{\loosestrategy}} (st41)
	(v5) edge node[right] {\textcolor{ForestGreen}{\loosestrategy}} (st51)
	;
	
	\draw[->,Blue, line width=0.5mm]
	(v2) edge node[above] 
	{\textcolor{Magenta}{\winstrategy}} node[below] 
	{\textcolor{ForestGreen}{\loosestrategy}} (st2)
	;
	
	\draw[->,Magenta, line width=0.5mm]
	(v1) edge node[below] {\textcolor{Magenta}{\winstrategy}} (st12)
	(v3) edge node[above] {\textcolor{Magenta}{\winstrategy}} (st32)
	(v4) edge node[above] {\textcolor{Magenta}{\winstrategy}} (st42)
	(v5) edge node[above] {\textcolor{Magenta}{\winstrategy}} (st52)
	;
	\end{tikzpicture}
	\caption{A reachability game where \winstrategy (depicted in pink) and 
		\loosestrategy (depicted in green) can be synchronised in 
		$S =\{\loc_3, \loc_4, \loc_5\}$ since if any \winstrategy-play from 
		\locinit reaches $S$.}
	\label{fig:ex_G-tree}
\end{figure}

\begin{example}
	\label{ex:synch-1}
	We consider \game be the reachability game depicted in 
	\figurename{~\ref{fig:ex_G-tree}}, and \looseplay be a non-winning MD-strategy 
	for \ReachPl depicted in green. By Lemma~\ref{lem:DHH-strat_find}, 
	we compute \winstrategy the winning MD-strategy depicted in pink that minimises 
	$\dstrat{\winstrategy}(\loosestrategy)$. However, 
	$\dstrat{\loosestrategy}(\winstrategy) = 2$ since the unique 
	\loosestrategy-play is $(\loc_1~\loc_3~\neg\LocsT)$. Now, consider 
	$\winstrategy^1$ define such that for all $\loc \in \LocsReach$
	\begin{displaymath}
	\winstrategy^1(\loc) = 
	\begin{cases}
	\winstrategy(\loc) & \text{if $\loc = \loc_1$;}\\
	\loosestrategy(\loc) & \text{otherwise.}
	\end{cases}
	\end{displaymath}
	This MD-strategy remains a winning MD-strategy in \game like the unique 
	$\winstrategy^1$-play is $(\loc_1~\loc_2~\LocsT)$ and minimises 
	$\dstrat{\winstrategy^1}(\loosestrategy)$. Moreover, 
	$\dstrat{\loosestrategy}(\winstrategy^1) = 1$ since the last choice of 
	$(\loc_1~\loc_3~\neg\LocsT)$ become a choice of $\winstrategy^1$. In 
	particular, this winning MD-strategy minimises $\dHH$.
	\markend
\end{example}

To synchronise \winstrategy with \loosestrategy in every needed vertex, we 
compute the connected components of the underlying undirected graph of 
$\game^{\winstrategy}$ and we synchronise all vertices that are not in the connected 
component of the initial vertex. More formally, let $S \subseteq \Locs$ be the 
set of vertices in the connected component of \locinit in the underlying 
undirected graph of $\game^{\winstrategy}$. From $S$ and \winstrategy, we define 
$\winstrategy^1$ such that for all vertices $\loc \in \LocsReach$, we have
\begin{displaymath}
\winstrategy^1(\loc) = 
\begin{cases}
\winstrategy(\loc) & \text{if $\loc \in S$;} \\
\loosestrategy(\loc) & \text{otherwise.}
\end{cases}
\end{displaymath}
Such a strategy is defined in Example~\ref{ex:synch-1}. To conclude this step, we 
need to prove that $\winstrategy^1$ remains optimal in $\widetilde{\game}$ and 
$\winstrategy^1$ minimises the distance $\dstrat{\loosestrategy}(\strategy)$. 

\begin{remark}
	In general case, $\winstrategy^1$ decreases the distance 
	$\dstrat{\loosestrategy}(\strategy)$, i.e. 
	$\dstrat{\loosestrategy}(\winstrategy^1) \leq  
	\dstrat{\loosestrategy}(\winstrategy)$. However, $\winstrategy^1$ does not 
	minimises $\dHH$: in the reachability game depicted in the left of 
	\figurename{~\ref{fig:ex_G-computation}}, $\winstrategy = \winstrategy^1$ and does not 
	minimise $\dHH$ as explained in Example~\ref{ex:synch-1}.
\end{remark}

First, we prove that $\winstrategy^1$ is still an optimal MD-strategy in 
$\widetilde{\game}$, i.e. it minimises the distance 
$\dstrat{\loosestrategy}(\strategy)$.

\begin{lemma}
	\label{lem:DHH-strat_transformation1}
	$\winstrategy^1$ is an optimal MD-strategy in $\widetilde{\game}$.
\end{lemma}
\begin{proof}
	We want to prove that all $\winstrategy^1$-plays are 
	also a \winstrategy-play. We reason by contraction and we suppose that 
	\winplay is a $\winstrategy^1$-play but not a \winstrategy-play. In 
	particular, there exists a vertex $\loc \in \LocsReach$ such that 
	$\winstrategy^1(\loc) \neq \winstrategy(\loc)$ (and without lost of 
	generality, we take the first one). By definition of $\winstrategy^1$, we 
	know that \loc is not in the connected component that contains 
	$\locinit$ (otherwise, both decision are equal). However, the prefix of 
	\winplay until \loc is also a \winstrategy-play (like we have chosen 
	\loc like the first vertex along \winplay such that both strategies are 
	distinct). In particular, \loc is in the connect component of $\locinit$ 
	(as there exists a path between $\locinit$ and \loc in the non-directed 
	graph underlying by $\game^{\winstrategy}$). We obtain a contradiction, and 
	like all $\winstrategy^1$-plays are also a \winstrategy-play, we 
	deduce that $\sCost(\winstrategy^1) \leq \sCost(\winstrategy)$. We conclude 
	by optimality of \winstrategy.
\end{proof}

To conclude the proof of Theorem~\ref{thm:find-trans_poly}, we want to prove 
that $\winstrategy^1$ minimises $\dstrat{\loosestrategy}(\strategy)$, when 
we suppose that $\game^{\loosestrategy}$ is acyclic. In particular, we 
obtain the following Proposition.

\begin{proposition}
	\label{prop:dHH-find_synch}
	Let \game be a reachability game and \loosestrategy be a non-winning 
	MD-strategy for \ReachPl such that $\game^{\winstrategy}$ is acyclic. Then, 
	$\winstrategy^1$ minimises $\dstrat{\loosestrategy}(\strategy)$.
\end{proposition}

Before proving this result for all reachability games that satisfy the hypothesis 
with the non-winning MD-strategy, we start by proving that 
$\dstrat{\loosestrategy}(\winstrategy^1) = 1$ when the reachability game is a 
tree. Intuitively, if \game is a tree when \loosestrategy and \winstrategy make 
distinct choices, then the play reaches a vertex that is not in $S$ (the 
connected component of \locinit). So, from this vertex, \winstrategy and 
\loosestrategy are synchronised. However, when \game is not a tree (even if it is 
acyclic), $\dstrat{\loosestrategy}(\winstrategy^1) = 1$ may be greater than $1$ 
(see \figurename{~\ref{fig:ex_G-dag}}).

\begin{lemma}
	\label{lem:dHH-tree}
	If \game is a tree, then $\dstrat{\loosestrategy}(\winstrategy^1) = 1$. 
\end{lemma}
\begin{proof}
	Let \looseplay be a \loosestrategy-play. If \looseplay is also a 
	$\winstrategy^1$-play then $\dist(\looseplay, \winstrategy^1) = 0$. Now, we 
	suppose that \looseplay is not a $\winstrategy^1$-play (exists, else 
	\loosestrategy is a winning strategy). To conclude the proof, we want to 
	prove that there exists exactly one vertex of \ReachPl along 
	\looseplay such that $\loosestrategy(\loc) \neq \winstrategy^1(\loc)$. 
	Since \looseplay is not a $\winstrategy^1$-play, there exists $\loc_i \in 
	\LocsReach$ a vertex of \looseplay such that $\loosestrategy(\loc_i) 
	\neq \winstrategy^1(\loc_i)$. By contradiction, we suppose that there exists 
	$\loc_j \neq \loc_i$ be a vertex of \ReachPl along \looseplay such that
	$\loosestrategy(\loc_j) \neq \winstrategy^1(\loc_j)$. Without lost of 
	generality, we suppose that $j > i$. By definition of $\winstrategy^1$, there
	exists a play in $\game^{\winstrategy}$ between $\locinit$ and $\loc_j$ 
	without reaching $\loc_i$. Otherwise, in $\loc_i$, we have 
	$\loosestrategy(\loc_i) = \winstrategy^1(\loc_i)$ (since this play between 
	$\locinit$ and $\loc_j$ is the prefix of \looseplay). However, since 
	\game is a tree, there exists only one play in \game between $\locinit$ 
	and $\loc_j$. We obtain a contradiction. 
\end{proof}

Even if \game is acyclic, we can not prove that, along all \loosestrategy-plays, 
there is at most one vertex such that $\winstrategy^1$ and \loosestrategy are 
distinct. In particular, we need to directly prove that the minimality of the 
distance $\dstrat{\loosestrategy}(\winstrategy)$ is reached by $\winstrategy^1$. 
Intuitively, the minimality of this distance is given by remarking that 
$\winstrategy^1$ synchronises most vertices as possible with \loosestrategy 
without breaking its optimality in $\widetilde{\game}$.

\begin{figure}[tbp]
	\centering
	\begin{tikzpicture}[xscale=.7,every node/.style={font=\footnotesize}, 
	every label/.style={font=\scriptsize}]
	\node[PlayerReach] at (0,0) (v1) {$\loc_1$};
	\node[PlayerReach] at (3, 1) (v2) {$\loc_2$};
	\node[PlayerSafe] at (3,-1) (v3) {$\loc_3$};
	\node[PlayerReach] at (6, -1) (v4) {$\loc_4$};
	\node[target] at (6,1) (t1) {$\locT$};
	\node[target] at (0,-1) (t2) {$\locNegT$};
	\node[target] at (9,-1) (t4) {$\locNegT$};
	
	\node[strat] at (2,-.7) (st11) {};
	\node[strat] at (2,.7) (st12) {};
	\node[strat] at (5, 1) (st2) {};
	\node[strat] at (1.5,-1) (st31) {};
	\node[strat] at (5,-1) (st32) {};
	\node[strat] at (6,.5) (st41) {};
	\node[strat] at (8,-1) (st42) {};
	
	% Connect the states with arrows
	\draw[->] 
	(v1) edge (v2)
	(v1) edge (v3)
	(v2) edge (t1)
	(v3) edge (v4)
	(v3) edge (t2)
	(v4) edge (t1)
	(v4) edge (t4)
	;
	
	\draw[->,ForestGreen, line width=0.5mm]
	(v1) edge node[above] {\textcolor{ForestGreen}{\loosestrategy}} (st11)
	(v3) edge node[below] {\textcolor{ForestGreen}{\loosestrategy}} (st31)
	(v4) edge node[above] {\textcolor{ForestGreen}{\loosestrategy}} (st42)
	;
	
	\draw[->,Blue, line width=0.5mm]
	(v2) edge node[above] 
	{\textcolor{Magenta}{\winstrategy}} node[below] 
	{\textcolor{ForestGreen}{\loosestrategy}} (st2)
	;
	
	\draw[->,Magenta, line width=0.5mm]
	(v1) edge node[below] {\textcolor{Magenta}{\winstrategy}} (st12)
	(v3) edge node[above] {\textcolor{Magenta}{\winstrategy}} (st32)
	(v4) edge node[right] {\textcolor{Magenta}{\winstrategy}} (st41);
	\end{tikzpicture}
	\caption{An acyclic reachability game where $\winstrategy^1$ depicted in pink 
		minimises $\dstrat{\loosestrategy}(\winstrategy^1)$ but not 
		$\dstrat{\loosestrategy}(\winstrategy^1) = 2$.}
	\label{fig:ex_G-dag}
\end{figure}
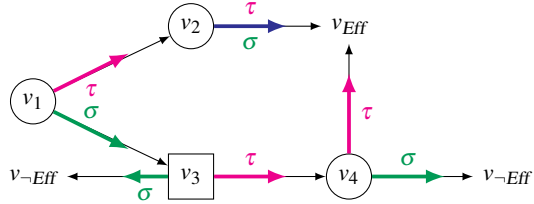

\begin{lemma}
	\label{lem:dHH-game}
	Let \game be a reachability game with \loosestrategy be a non-winning 
	MD-strategy such that $\game^{\loosestrategy}$ is acyclic, then 
	$\winstrategy^1$ minimises $\dstrat{\loosestrategy}(\strategy)$, 
	$\dstrat{\loosestrategy}(\winstrategy^1) = 
	\min\{\dstrat{\loosestrategy}(\strategy) \mid \winstrategy 
	\text{ a winning MD-strategy}\}$. 
\end{lemma}
\begin{proof}
	We reason by contradiction and we suppose that there exists a 
	winning MD-strategy \strategy such that $ \dstrat{\loosestrategy}(\strategy) 
	< \dstrat{\loosestrategy}(\winstrategy^1)$. In particular, by 
	rewriting distances $\dstrat{\loosestrategy}(\strategy)$, we suppose that 
	\begin{displaymath}
	\dstrat{\loosestrategy}(\strategy) = 
	\sup_{\substack{\looseplay \\ \text{\loosestrategy-play}}} 
	\dist(\looseplay, \strategy) <  
	\sup_{\substack{\looseplay \\ \text{\loosestrategy-play}}} 
	\dist(\looseplay, \winstrategy^1) = \dstrat{\loosestrategy}(\winstrategy^1)
	\end{displaymath}
	Thus, we deduce that there exists a \loosestrategy-play, denoted \looseplay, 
	such that $\dist(\looseplay, \strategy) < \dist(\looseplay, \winstrategy^1)$. 
	Moreover, along this play, there exists $\loc \in \looseplay$ such that 
	$\strategy(\loc) = \loosestrategy(\loc)$ and $\loosestrategy(\loc) \neq 
	\winstrategy^1(\loc)$. Let $\loc$ be the first such vertex along \looseplay. 
	In particular, we can decomposed \looseplay as 
	$\looseplay = \looseplay_1 \loc \looseplay_2$. 
	
	We suppose that along $\looseplay_2$, all vertices $\loc' \in S$ (be the 
	connected component of \locinit). In particular, by definition of 
	$\winstrategy^1$,$\looseplay_2$ is a $\winstrategy^1$-play, thus 
	$\dist(\looseplay, \winstrategy^1) = 1$. In particular, by hypothesis 
	$\dist(\looseplay, \strategy) = 0$ which is only possible when \looseplay 
	reaches the effect. In this case $\looseplay$ be a \strategy-play that 
	contradict the optimality of $\winstrategy^1$ in $\widetilde{\game}$. 
	Otherwise, we can define another \play from \loc that is a \strategy-play 
	and a \loosestrategy-play that does not reach the target that contradict 
	the fact that \strategy is a winning MD-strategy.
	
	Now, we suppose that there exists $\loc'$ in \looseplay such that 
	\looseplay reaches $\loc'$ after \loc and $\loc' \in S$. In this case, we 
	define $\strategy'$ be a strategy such that for all plays $\play'$, we have 
	\begin{displaymath}
	\strategy'(\play') = 
	\begin{cases}
	\winstrategy^1(\play') & \text{if $\play'$ is not a prefix of \looseplay;} \\
	\winstrategy^1(\play') & \text{if $\play'$ is a prefix of $\looseplay_1$, or 
		is is a prefix of $\looseplay_1 \loc \looseplay_2$ with $\loc' \in 
		\looseplay_2$;} \\
	\loosestrategy(\play') & \text{otherwise.}
	\end{cases}
	\end{displaymath}
	This strategy is winning since \strategy is a winning MD-strategy and 
	$\sCost(\strategy) < \sCost(\winstrategy^1)$ that contradict the 
	optimality of $\winstrategy^1$ in $\widetilde{\game}$.
\end{proof}

\end{document}